\documentclass[
  reprint,
  floatfix,
  aps,
  pra,
  nofootinbib,
  eqsecnum,
  longbibliography
]{revtex4-2}

\usepackage{graphicx}
\PassOptionsToPackage{dvipsnames}{xcolor}
\usepackage[most]{tcolorbox}
\let\paperTcolorboxEnvironment\tcolorboxenvironment
\usepackage{Lucky_preamble}
\let\tcolorboxenvironment\paperTcolorboxEnvironment

\AtEndEnvironment{proposition}{\setcounter{mpfootnote}{\value{footnote}}}

\usepackage{Menicucci_preamble}
\newcommand{\criterion}[1]{\ensuremath{(#1)}}
\newcommand{\pointwise}{\mathbin{\vcenter{\hbox{\scalebox{2}{$\cdot$}}}}}
\newtcolorbox{notation}{enhanced,breakable,parbox=false,colback=rembg,colframe=remframe,boxrule=0pt,arc=0mm,left=3mm,right=2mm,top=1mm,bottom=1mm,before skip=10pt,after skip=10pt}
\newcommand{\indicDKS}{\Pi^{\mathrm c}}

\begin{document}

\title{Complete parameterization and parameter-space topology of discrete Wigner representations on \(d\times d\) phase space}

\author{Lucky K. Antonopoulos} 
\email{Lucky.K.Antonopoulos@gmail.com}
\affiliation{Centre for Quantum Computer Performance and Integration, Department of Physics, RMIT University, Melbourne, Victoria 3000, Australia}

\author{Nicolas C. Menicucci}
\email{nicolas.menicucci@rmit.edu.au}
\affiliation{Centre for Quantum Computer Performance and Integration, Department of Physics, RMIT University, Melbourne, Victoria 3000, Australia}

\begin{abstract}
Toroidal discrete Wigner functions (DWFs) for finite-dimensional systems are non-unique. We completely parameterize and topologically classify the labeled single-qudit ${d\times d}$ representations whose phase-point operators are Hermitian, unit-trace, Hilbert--Schmidt orthogonal, and Weyl--Heisenberg covariant, determining the independent parameter freedom in every dimension. Our previously established stencil theorem expresses this family as descents of a doubled ${2d\times2d}$ parent Wigner representation. Here, we solve its projected-stencil admissibility conditions explicitly to obtain the full parameter space of this family. In the symplectic-Fourier representation, admissibility fixes the modulus, leaving phase data on a ${d\times d}$ base cell satisfying a parity-dependent twisted-oddness relation. This gives the parameter space ${(S^1)^{(d^2-1)/2}}$ for odd ${d}$ and ${(S^1)^{(d^2-4)/2}\times\intsT[2][3]}$ for even ${d}$. Canonical horizontal and vertical marginals reduce these spaces to ${(S^1)^{(d-1)^2/2}}$ and ${(S^1)^{d\,(d-2)/2}\times\intsT[2]}$, respectively. The same phase data parameterize validity-preserving rephasings of the doubled Weyl--Heisenberg displacement operators. Requiring these operators to have order dividing the Hilbert-space dimension ${d}$ replaces each ${S^1}$ factor by a discrete ${\intsT[d]}$ factor. The associated symplectic-Fourier characteristic functions encode the same operator information, with pointwise magnitudes independent of the valid stencil convention and hence providing convention-independent information about the represented operator within this family. This classification separates the freedom intrinsic to DWF validity from that selected by marginal and displacement-algebra requirements and makes explicit the structural distinction between odd and even dimensions.
\end{abstract}

\maketitle

\date{}

\section{Introduction}
Discrete Wigner functions (DWFs) characterize states and operators within a phase space. These finite-dimensional counterparts of the Wigner representation~\cite{wigner_Quantum_1932} are non-unique: different constructions depend on dimension, phase convention, covariance requirements and marginal structure~\cite{wootters_Wignerfunction_1987,chaturvedi_Wigner_2010,miquel_Quantum_2002,ferrie_Quasiprobability_2011,leonhardt_Discrete_1996,horibe_Existence_2002,raussendorf_Contextuality_2017,zak_Doubling_2011,gibbons_Discrete_2004,gross_Hudsons_2006,agam_Semiclassical_1995,bianucci_Discrete_2002}. Our previous stencil theorem~\cite{antonopoulos_Grand_2025} exhausts and unifies the single-qudit, toroidal ${d\times d}$ DWFs satisfying a set of basic criteria, \criterion{\mathrm{A}1}--\criterion{\mathrm{A}4}, that form a discrete analog of the Stratonovich--Weyl criteria~\cite{stratonovich_Distributions_1957,brif_General_1998}; we call such DWFs \emph{valid}. Its linear maps also enable representation-based comparisons.

A stencil is a function that specifies the weights used to combine values of a doubled ${2d\times2d}$ parent Wigner representation. Cross-correlating the stencil with that representation and evaluating the result at even lattice points produces a ${d\times d}$ phase-space function. The classification concerns labeled single-qudit phase-point operator (PPO) frames on a ${d\times d}$ phase space satisfying \criterion{\mathrm{A}1}--\criterion{\mathrm{A}4}: the PPOs are Hermitian, unit-trace, Hilbert--Schmidt orthogonal and covariant under Weyl--Heisenberg displacement operators (WHDOs). The stencil theorem provides necessary and sufficient criteria for a stencil to be valid and hence to generate a PPO frame with these properties. Our previous work~\cite{antonopoulos_Grand_2025} reduced the classification problem to these stencil conditions and established the projected-stencil description, in which each labeled representation is associated with a unique projected stencil. However, those admissibility conditions were not solved in terms of independent parameters, nor was the topology of the resulting family determined. In the present work, we solve this remaining projected-stencil classification problem explicitly. Projection retains only the part of each stencil that affects the resulting phase-space function, removing redundant stencil freedom associated with the four-fold redundancy of the doubled parent representation.

The broader literature contains many distinct finite-dimensional Wigner constructions with differing dimensional and structural requirements, as summarized, for example, by Ferrie and Emerson~\cite{ferrie_Framed_2009}. Takami \textit{et al.}~\cite{takami_Wigner_2001} parameterized the general family of DWF representations satisfying canonical position and momentum marginals, Hermiticity, and orthonormal completeness by real orthogonal matrices, without imposing WHDO covariance. Chaturvedi \textit{et al.}~\cite{chaturvedi_Wigner_2010} give a related discrete sign-valued construction within the canonically marginalized setting, which lies within the corresponding family classified here after matching conventions. Raussendorf \textit{et al.}~\cite{raussendorf_Role_2023} develop a more general Pauli-covariant operator-basis framework. In contrast to these earlier works, we give a new topological classification, together with a complete parameterization, of the full \criterion{\mathrm{A}1}--\criterion{\mathrm{A}4} family in every dimension.

The contributions in this paper are as follows. We take the valid DWF family characterized in Ref.~\cite{antonopoulos_Grand_2025} and give a complete parameterization of all representations satisfying \criterion{\mathrm{A}1}--\criterion{\mathrm{A}4}, together with the topology of the resulting parameter space and of the restricted families selected by additional conditions. For odd-dimensional Hilbert spaces, the valid family forms a single connected torus, whereas for even dimensions it separates into eight disconnected toroidal components distinguished by three discrete binary choices. This makes explicit a qualitative difference in the remaining DWF freedom: it is entirely continuous in odd dimensions, while even dimensions additionally possess distinct discrete sectors. We further determine the parameterizations of the subfamilies selected by canonical horizontal and vertical marginals and by an order-dividing-$d$ requirement on the associated stencil-induced WHDOs, showing that these familiar marginal and algebraic requirements constrain the representation without generally selecting a unique one. We show that the stencil-induced WHDOs arise as validity-preserving rephasings of the doubled WHDOs, so that the continuous and discrete freedom classified here remains even after fixing the Weyl commutation law. We also analyze the corresponding symplectic-Fourier characteristic functions constructed from these WHDOs, which encode the same operator information as the DWFs, and show that changing the valid stencil convention preserves their pointwise magnitudes and zeros while modifying only the phases of nonzero components. These preserved quantities therefore provide convention-independent information about the represented operator within the \criterion{\mathrm{A}1}--\criterion{\mathrm{A}4} family.

We begin in \Cref{sec_stencil_frame_symp_Four_admissibility} by recalling the projected-stencil framework and its validity criteria. \Cref{sec_cons_mod_phase_norm_form} expresses these criteria in terms of a constant-modulus phase normal form, and \Cref{sec_param_by_base_cell_phase_funcs} develops the complete base-cell parameterization in terms of twisted-odd phase data. In \Cref{sec_topo_and_margin_constr_subspaces}, we determine the topology, introduce the associated stencil-dependent WHDOs and characteristic functions and derive the restriction imposed by canonical horizontal and vertical marginals. \Cref{sec_induced_WHDO_struct_and_algeb_restr} develops examples, the induced WHDOs multiplication laws and the order-dividing-${d}$ restriction. Finally, \Cref{sec_disc_and_conc} discusses the classification, its relation to earlier constructions and its scope.

\section{Stencil framework and symplectic-Fourier admissibility}
\label{sec_stencil_frame_symp_Four_admissibility}
Our previous work~\cite{antonopoulos_Grand_2025} established necessary and sufficient conditions for constructing DWFs whose PPOs are Hermitian, unit-trace, Hilbert--Schmidt orthogonal and Weyl--Heisenberg covariant. We recall the stencil framework, beginning with the notation and operator conventions used throughout this paper.

\subsection{Recalled notation and doubled-space machinery}
\label{app_recalled_notation_and_doubled_space_machinery}
We work with two-component vectors of the form ${\bm{v}=(v_1,v_2)^\tp}$. For ${N\in\intsT[>0]}$, the discrete phase space is ${\mathcal{P}_{N}\coloneqq\intsT[N][2]}$, with phase-space points added componentwise modulo ${N}$. We write ${a\equiv_T b}$ when ${a-b\in T\intsT}$, with vector congruences understood componentwise. Furthermore, we use ${\bm{m}\in\mathcal{P}_{2d}}$ and ${\bm{\alpha}\in\mathcal{P}_{d}}$ for the \emph{doubled} and \emph{undoubled} phase-space points, respectively, ${\bm{b}\in\intsT[2][2]}$ for a bit vector and ${\bm{\kappa}\in\intsT[][2]}$ for an integer displacement. 

On ${\mathcal{P}_N=\intsT[N][2]}$, addition and negation are componentwise modulo ${N}$; ${\{x\}_N\in\{0,\ldots,N-1\}}$ denotes the ordinary remainder, also applied componentwise to vectors. For ${a\in\reals}$ and ${T>0}$, our exponential phase notation is
\begin{align}
    \omegaT[T]\coloneqq e^{2\pi i/T},\qquad
    \omegaT[T][a]\coloneqq e^{2\pi ia/T}.
    \label{def_exp_phase_notation}
\end{align}
For integer ${a}$ and positive integer ${T}$, the \emph{Kronecker comb} ${\DeltaT[T][a]}$ is ${1}$ when ${T\mid a}$ and ${0}$ otherwise, with ${\DeltaT[T][\bm a]=\DeltaT[T][a_1]\DeltaT[T][a_2]}$.

We use the symplectic form matrix
\begin{align}
    \mat\Omega
    &\coloneqq
    \bmat{0}{-1}{1}{0}.
    \label{eq_symp_form_matrix}
\end{align}
For ${T\in\intsT[>0]}$, the \emph{symplectic discrete Fourier transform} (SDFT) of ${f:\mathcal P_T\to\complex}$ is
\begin{align}
 \check f(\bm a')=\SDFT[T][f](\bm a')
 \coloneqq\frac1T\sum_{\bm a\in\mathcal P_T}f(\bm a)\omegaT[T][-\symprod{\bm a'}{\bm a}],
 \label{def_SDFT}
\end{align}
with ${\bm a'\in\mathcal P_T}$, using the conventions of Ref.~\cite{antonopoulos_Grand_2025}.
The SDFT is self-inverse and acts entrywise on operator-valued functions. For ${f,g:\mathcal P_T\to\complex}$, their \emph{discrete cross-correlation} is
\begin{align}
    \crossc{f}{g}(\bm a)
    \coloneqq\sum_{\bm c\in\mathcal P_T}f(\bm c)^*g(\bm c+\bm a),
    \label{def_disc_cross_corr}
\end{align}
with arguments reduced modulo ${T}$.

The \emph{position basis} (also known as the computational basis) ${\{\ket j_q\}_{j\in\intsT[d]}}$ and its Fourier-dual \emph{momentum basis} (also known as the Fourier basis) satisfy
\begin{align}
    \ket k_p
    \coloneqq
    \frac1{\sqrt d}\sum_{j\in\intsT[d]}\omegaT[d][jk]\ket j_q.
    \label{def_mom_pos_bases}
\end{align}
The \emph{clock} and \emph{shift} operators obey ${\wzed\ket j_q=\omegaT[d][j]\ket j_q}$ and ${\wex\ket j_q=\ket{j+1}_q}$, respectively, where ket labels are modulo ${d}$. The \emph{general-form Weyl--Heisenberg displacement operators} (WHDOs%
\footnote{We pronounce these as ``who-doos''.}%
) are
\begin{align}
    \hat V^{\mathrm{Gen}}(\bm\kappa)
    \coloneqq\Phi(\bm\kappa;d)\wzed^{\kappa_2}\wex^{\kappa_1},
    \label{def_WHDO_general_phase}
\end{align}
for ${\Phi(\bm\kappa;d)\in U(1)}$. We also define the \emph{Hilbert--Schmidt inner product} for operators ${\hat{O},\hat{O}'\in\mathcal L(\complex^d)}$ by ${\HSip{O}{O'}\coloneqq\traceOf{\hat{O}^\dagger}{\hat{O}'}}$.

For a qudit Hilbert space ${\mathcal H\cong\complex^d}$, a \emph{${d\times d}$ PPO frame} is a family of PPOs ${\{\hat{A}(\bm{\alpha})\}_{\bm{\alpha}\in\mathcal{P}_{d}}}$ that spans $\mathcal{L}(\complex^{d})$. Since this family contains $d^2$ operators and $\dim\mathcal{L}(\complex^{d})=d^2$, the PPOs form an informationally complete operator basis. Validity of such a frame imposes the following additional requirements.
\begin{definition}[Valid PPO frame]
\label{def_DV_valid_PPO_frame}
    A \emph{valid PPO frame} is a set of ${d^2}$ discrete PPOs with ${N=d}$, namely ${\{\hat{A} (\bm{\alpha}) \mid \bm{\alpha} \in \mathcal{P}_{d}\}}$, forming a Hermitian, unit-trace, Hilbert--Schmidt orthogonal and WHDO-covariant operator basis. Explicitly, for all ${\bm{\alpha}, \bm{\beta} \in \mathcal{P}_{d}}$, the following hold:
    \begin{itemize}
    \item[\criterion{\mathrm{A}1}.]\label{A1} (Hermiticity/reality). For all ${\bm{\alpha}}$, the phase-point operators are Hermitian:
    ${\hat{A}(\bm{\alpha}) = \hat{A}(\bm{\alpha})^\dagger.}$
    
    \item[\criterion{\mathrm{A}2}.]\label{A2} (Unit trace). For all ${\bm{\alpha}}$, the PPOs have unit trace:
    ${\traceOf{\hat{A}(\bm{\alpha})}{} = 1.}$
    
    \item[\criterion{\mathrm{A}3}.]\label{A3} (Hilbert--Schmidt orthogonality). For all ${\bm{\alpha},\bm{\beta}}$, the PPOs are Hilbert--Schmidt orthogonal:
    \begin{align}
        \HSip{A(\bm{\alpha})}{A(\bm{\beta})}
        =
        \traceOf{\hat{A}(\bm{\alpha})^\dagger}{\hat{A}(\bm{\beta})}
        =
        d\,\DeltaT[d][\bm{\alpha} - \bm{\beta}].
    \end{align}
    
    \item[\criterion{\mathrm{A}4}.]\label{A4} (WHDO covariance). For all ${\bm{\alpha}}$ and ${\bm{\kappa}\in\intsT[][2]}$, the general-form discrete WHDOs act on PPO labels by discrete phase-space translation under conjugation:
    \begin{align}
        \hat{V}^{\mathrm{Gen}}(\bm{\kappa}) \hat{A}(\bm{\alpha}) \hat{V}^{\mathrm{Gen}}(\bm{\kappa})^\dagger = \hat{A}(\bm{\alpha} + \bm{\kappa}).
    \end{align}
    Here, ${\bm{\alpha}+\bm{\kappa}}$ is understood as addition in ${\mathcal{P}_{d}}$ (componentwise modulo ${d}$) and the general-form WHDOs are given in \Cref{def_WHDO_general_phase}.
\end{itemize}
\end{definition}
Since ${\Phi(\bm\kappa;d)}$ is a scalar phase, \criterion{\mathrm{A}4} is independent of the WHDO phase convention. The associated \emph{valid DWF} representing a discrete operator ${\hat O\in\mathcal{L}(\mathcal H)}$ is the trace representation 

\begin{align}
    W_{\hat O}(\bm\alpha)
    \coloneqq
    \frac{1}{d}
    \traceOf{\hat A(\bm\alpha)^\dagger}{\hat O},
\end{align}
where ${\bm\alpha\in\mathcal P_d}$. The ${d^2}$ orthogonal PPOs form an informationally complete operator basis, so this representation is \emph{faithful}: the represented operator ${\hat O}$ can be reconstructed uniquely from its DWF values.

The discrete-variable \emph{doubled WHDOs} of Ref.~\cite{antonopoulos_Grand_2025}, in the notation used here, are
\begin{align}
 \hat V^{(2d)}(\bm m)
 \coloneqq\omegaT[2d][-m_1m_2]\wzed^{m_2}\wex^{m_1}.
 \label{def_DV_doubled_WHDOs}
\end{align}
Doubling concerns the label lattice; the Hilbert-space dimension remains ${d}$. With the \emph{discrete parity operator} ${\hat R\coloneqq\sum_{j\in\intsT[d]}\qoutprod{-j}{j}}$, the \emph{doubled PPOs} are
\begin{align}
    \hat A^{(2d)}(\bm m)\coloneqq\hat V^{(2d)}(\bm m)\hat R.
    \label{eq_doubled_PPO}
\end{align}
The doubled PPO covariance recalled from Ref.~\cite{antonopoulos_Grand_2025} is
\begin{align}
    \hat V^{\mathrm{Gen}}(\bm\kappa)\hat A^{(2d)}(\bm m)
    \hat V^{\mathrm{Gen}}(\bm\kappa)^\dagger
    =
    \hat A^{(2d)}(\bm m+2\bm\kappa),
    \label{prop_doubled_PPO_covariance}
\end{align}
where the right-hand label is reduced modulo ${2d}$. Thus, physical translations act on doubled PPO labels by even shifts.

Finally, we define the associated doubled DWF through the following doubled Wigner and Weyl maps. Let ${\{\hat{A}^{(2d)}(\bm{m})\}_{\bm{m}\in\mathcal{P}_{2d}}}$ be the doubled PPO self-dual frame, where \emph{self-dual} means that the frame is its own dual under the Hilbert--Schmidt inner product, up to an overall normalization factor. For ${\hat O\in\mathcal L(\complex^d)}$ and ${f:\mathcal P_{2d}\to\complex}$, the \emph{doubled Wigner map} ${\Wig^{(2d)}}$ and \emph{doubled Weyl map} ${\Op^{(2d)}}$ are defined by
\begin{align}
    \bigl(\Wig^{(2d)}[\hat O]\bigr)(\bm m) 
    &\coloneqq 
    \frac{1}{2d}\traceOf{\hat A^{(2d)}(\bm m)^\dagger}{\hat O},
    \label{eq_doubled_Wigner_map_recall}\\
    \Op^{(2d)}[f] 
    &\coloneqq 
    \frac{1}{2}\sum_{\bm m\in\mathcal P_{2d}}f(\bm m)\hat A^{(2d)}(\bm m),
    \label{eq_doubled_Weyl_map_recall}
\end{align}
and these maps satisfy
\begin{align}
    \Op^{(2d)}\circ\Wig^{(2d)}
    =
    \mathrm{id}_{\mathcal{L}(\complex^d)},
    \label{eq_doubled_Op_Wig_composition}
\end{align}
where ${\mathrm{id}_{\mathcal{L}(\complex^d)}}$ is the identity map on the operator space ${\mathcal{L}(\complex^d)}$. The \emph{doubled DWF} representing ${\hat O}$ is then the function 
\begin{align}
    W_{\hat O}^{(2d)}\coloneqq\Wig^{(2d)}[\hat O].
    \label{eq_doubled_DWF}
\end{align}
This function is defined on the doubled discrete phase space ${\mathcal{P}_{2d}}$, and is the parent function for all valid ${d\times d}$ DWFs and the starting point for the stencil-based framework of Ref.~\cite{antonopoulos_Grand_2025}.

\subsection{Projected stencils and admissibility}
\label{sec_proj_stencils_and_admissibility}
The doubled-representation space ${\complex^{\mathcal{P}_{2d}}}$ has dimension ${4d^2}$, whereas ${\mathcal L(\complex^d)}$ has dimension ${d^2}$. Hence, a generic doubled-lattice function cannot be interpreted as independent qudit phase-space data. The doubled Wigner map fills only a distinguished ${d^2}$-dimensional subspace, on which the representation is faithful~\cite{antonopoulos_Grand_2025}. The following projection isolates this subspace, separating operator information from redundant doubled-lattice data.

\begin{definition}[Doubled-space projection map]
\label{def_doubled_space_projection_map}
    The \emph{doubled-space projection map} is the linear map
    \begin{align}
        \Proj \coloneqq \Wig^{(2d)} \circ \Op^{(2d)} : \complex^{\mathcal{P}_{2d}} \to \complex^{\mathcal{P}_{2d}},
    \end{align}
    which acts on a function ${f : \mathcal{P}_{2d} \to \complex}$ by
    \begin{align}
        \bar{f} \coloneqq \Proj f
        =
        \Wig^{(2d)} \bigl[ \Op^{(2d)}[f] \bigr]
        =
        W^{(2d)}_{\Op^{(2d)}[f]}.
    \end{align}
\end{definition}
Here ${\im(\Proj)}$ denotes the \emph{image} of ${\Proj}$.
Furthermore, the doubled-space projection acts explicitly as
\begin{align}
    \bar{f}(\bm{m})
    &=
    \frac{1}{4}
    \sum_{\bm{b}\in\intsT[2][2]}
    (-1)^{\symprod{\bm{b}}{\bm{m}}-b_1b_2d}
    f(\bm{m}-d\,\bm{b}).
    \label{eq_proj_explicit}
\end{align}
Thus, one ${d\times d}$ base cell determines the remaining three blocks pointwise, up to the prescribed signs. The resulting structure is illustrated for ${d=2}$ and ${d=3}$ in \Cref{fig_doubled_DWF_redundancy}.

\begin{figure}[h]
    \centering
    \includegraphics[width=\columnwidth]{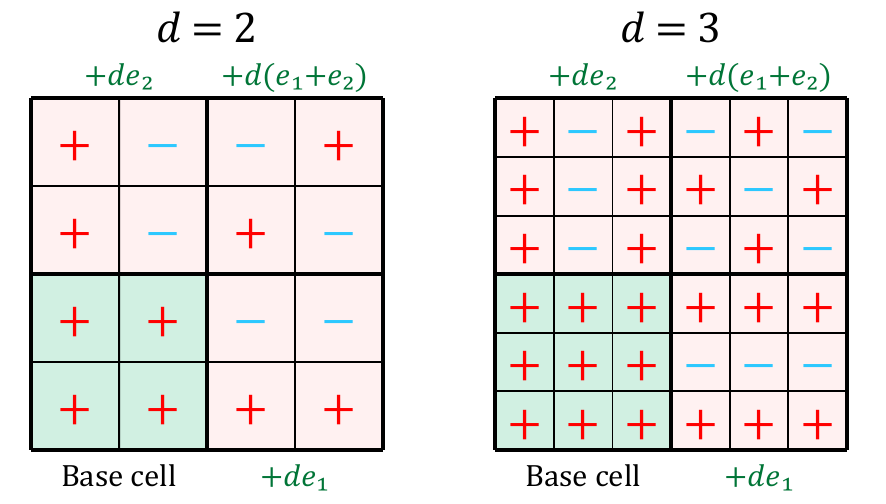}
    \caption[Fourfold redundancy of functions in ${\im(\Proj)}$]
    {Fourfold redundancy of functions in ${\im(\Proj)}$ for ${d=2}$ (left) and ${d=3}$ (right). The green highlighted ${d\times d}$ block is a base cell; the other three red highlighted blocks are fixed pointwise by the projection sign structure in \Cref{eq_proj_explicit}. The block labels indicate translations by ${d\,e_1}$, ${d\,e_2}$, and ${d\,(e_1+e_2)}$, with ${e_1=(1,0)^\tp}$ and ${e_2=(0,1)^\tp}$. The ${d\,b_1b_2}$ term in the sign exponent of \Cref{eq_proj_explicit} is even for ${d=2}$, but contributes an extra minus sign to the doubly shifted block for ${d=3}$. The symbols show the resulting sign pattern, with positive base-cell values chosen for illustration.}
    \label{fig_doubled_DWF_redundancy}
\end{figure}

By \Cref{eq_doubled_Op_Wig_composition}, ${\Proj^2=\Proj}$, ${\Proj\circ\Wig^{(2d)}=\Wig^{(2d)}}$ and ${\Op^{(2d)}\circ\Proj=\Op^{(2d)}}$. Consequently, ${\im(\Proj)=\im(\Wig^{(2d)})}$. Thus, ${\Proj}$ removes redundant doubled-lattice data that do not affect the represented operator. A \emph{stencil} is any complex-valued function ${M:\mathcal P_{2d}\to\complex}$. To descend to a ${d\times d}$ phase space, we cross-correlate [\Cref{def_disc_cross_corr}] the doubled DWF with a stencil and restrict to the even sites ${2\bm\alpha}$.

\begin{definition}[${M}$-DWF and ${M}$-PPO]
    \label{def_M_DWF_and_M_PPO}
    Let ${M:\mathcal{P}_{2d} \to \complex}$ be a stencil and let ${\hat{O}\in\mathcal{L}(\complex^d)}$ be a qudit operator. The \emph{${M}$-DWF} generated by ${M}$ (of the operator ${\hat{O}}$) is defined by
    \begin{align}
        W_{\hat{O}}^{M}(\bm{\alpha})
        \coloneqq
        \crossc{M}{W_{\hat{O}}^{(2d)}}(2\bm{\alpha}),
        \qquad
        \bm{\alpha}\in\mathcal{P}_{d},
    \label{eq_M_DWF_def}
    \end{align}
    where the cross-correlation ${\star}$ of \Cref{def_disc_cross_corr} is taken over ${\mathcal{P}_{2d}}$ (i.e., with $T=2d$). The \emph{${M}$-PPO} generated by ${M}$ is given by
    \begin{align}
        \hat{A}^{M}(\bm{\alpha})
        \coloneqq
        \frac{1}{2} \crossc{M^{*}}{\hat{A}^{(2d)}}(2\bm{\alpha}).
    \label{eq_M_PPO_def}
    \end{align}
    Equivalently, the induced ${M}$-DWF may be written in trace form as
    \begin{align}
        W_{\hat{O}}^{M}(\bm{\alpha})
        =
        \frac{1}{d}\traceOf{\hat{A}^{M}(\bm{\alpha})^\dagger}{\hat{O}}.
    \label{eq_M_DWF_trace_form}
    \end{align}
\end{definition}

The induced ${M}$-PPO depends only on the projected stencil:
\begin{align}
    \hat{A}^{M}(\bm{\alpha})
    &=
    \Op^{(2d)}
    \left[M(\inputdot-2\bm{\alpha})\right]
    \nonumber\\
    &=
    \Op^{(2d)}
    \left[\bar{M}(\inputdot-2\bm{\alpha})\right].
    \label{eq_projected_stencil_generates_same_PPO}
\end{align}
Here we used ${\Op^{(2d)}\circ\Proj=\Op^{(2d)}}$ and the fact that ${\Proj}$ commutes with translations by even labels, as follows from \Cref{eq_proj_explicit}.

General stencils need not yield valid DWFs. The natural question therefore is: which functions ${M\in\complex^{\mathcal P_{2d}}}$ generate valid ${d\times d}$ DWFs? WHDO covariance \criterion{\mathrm{A}4} is automatic: cross-correlation is linear in the doubled PPO and evaluation at ${2\bm\alpha}$ converts doubled-label covariance into translation on ${\mathcal P_d}$. The remaining conditions correspond to \criterion{\mathrm{A}1}--\criterion{\mathrm{A}3}: Hermiticity, unit trace and Hilbert--Schmidt orthogonality.

A stencil ${M}$ is \emph{valid} when its projection ${\bar M=\Proj M}$ satisfies the reality, normalization and autocorrelation criteria of Ref.~\cite{antonopoulos_Grand_2025}:
\begin{align}
 \criterion{\mathrm{M}1}\quad &\bar M(\bm m)^*=\bar M(\bm m),
       &&\bm m\in\mathcal P_{2d},\nonumber\\
 \criterion{\mathrm{M}2}\quad &\sum_{\bm m\in\mathcal P_{2d}}\bar M(\bm m)=1,\nonumber\\
 \criterion{\mathrm{M}3}\quad &\crossc{\bar M}{\bar M}(2\bm\alpha)=\DeltaT[d][\bm\alpha],
       &&\bm\alpha\in\mathcal P_d.\nonumber
\end{align}
The Stencil Theorem of Ref.~\cite{antonopoulos_Grand_2025} establishes that every valid DWF over a ${d\times d}$ phase space arises as a valid-stencil descent of the doubled DWF and, conversely, that every valid stencil generates a valid ${M}$-DWF. Thus, the stencil construction exhausts the \criterion{\mathrm{A}1}--\criterion{\mathrm{A}4} class of valid ${d\times d}$ DWFs.

By Corollary~1 of Ref.~\cite{antonopoulos_Grand_2025}, two valid stencils generate the same labeled ${M}$-DWF representation for every operator exactly when ${\bar M_1=\bar M_2}$, or equivalently, when their difference is in the kernel of ${\Proj}$: ${M_1-M_2\in\ker\Proj}$. We call this \emph{projected stencil uniqueness}.
Consequently, ${\bar M}$ is the representation-relevant object: kernel differences encode redundant doubled-lattice data, not additional operator information. \Cref{fig_projected_stencil_nonuniqueness} illustrates this many-to-one relation.

\begin{figure*}[t]
    \centering
    \includegraphics[width=0.7\textwidth]{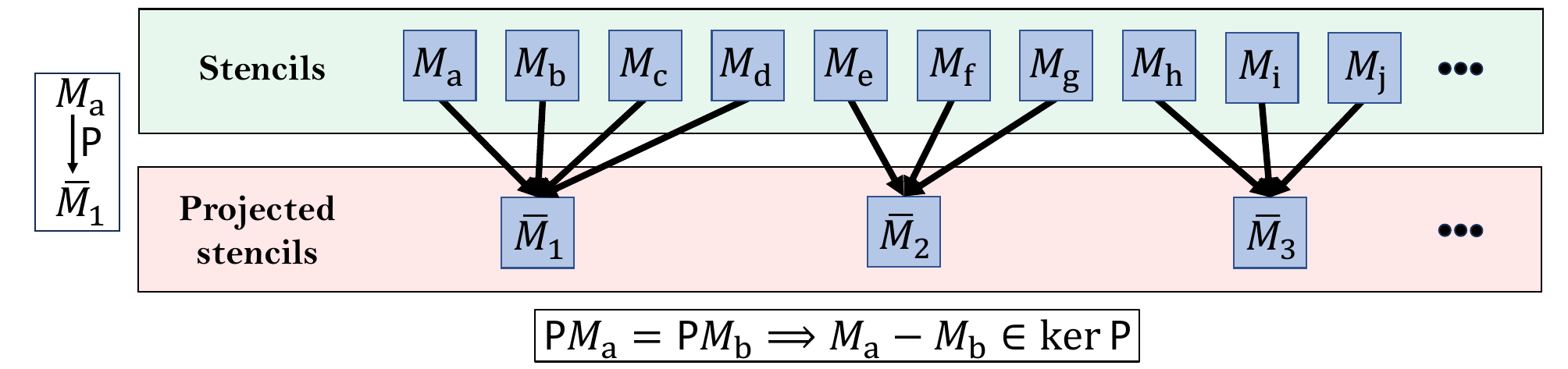}
    \caption{Projected stencil uniqueness under the doubled-space projection ${\Proj}$. Distinct unprojected stencils ${M_a}$ and ${M_b}$ may have the same projected representative ${\bar{M}_1}$, equivalently ${M_a-M_b\in\ker\Proj}$. For valid stencils, such kernel differences are redundant doubled-lattice data and yield the same ${M}$-DWF and ${M}$-PPO frame.}
    \label{fig_projected_stencil_nonuniqueness}
\end{figure*}

A further identity proved in Ref.~\cite{antonopoulos_Grand_2025} is that the doubled-space projection commutes with the ${2d}$-sized symplectic discrete Fourier transform:
\begin{align}
    \SDFT[2d]\circ\,\Proj = \Proj\circ\,\SDFT[2d].
    \label{eq_proj_commutes_with_SDFT}
\end{align}
Equivalently, for any ${f:\mathcal{P}_{2d}\to\complex}$, ${\check{\bar{f}} = \bar{\check{f}}}$, where the check accent (${\check{f}}$) denotes the symplectic discrete Fourier transform (SDFT) defined in \Cref{def_SDFT}. We use this relation to pass from the phase-space stencil criteria ${\criterion{\mathrm{M}1}}$--${\criterion{\mathrm{M}3}}$ to their equivalent symplectic-Fourier form.

The projected-stencil validity criteria have the following equivalent symplectic-Fourier form, established in Appendix Sec.~3 of the published Ref.~\cite{antonopoulos_Grand_2025}.

\begin{proposition}[Symplectic-Fourier projected-stencil validity criteria]
\label{prop_fourier_domain_stencil_validity}
    A stencil~${M}$ is valid if and only if the projected SDFT \({\bar{\check{M}} = \check{\bar{M}}}\) satisfies the following conditions:
    \begin{enumerate}[align=left]
    \item[\criterion{\check{\mathrm{M}}1}.] \emph{Hermitian symmetry:}
    \begin{align}
        \bar{\check{M}}(\bm{m})^* = \bar{\check{M}}(-\bm{m}), \quad \forall \bm{m}\in\mathcal{P}_{2d}. \label{checkM1}
    \end{align}
    \item[\criterion{\check{\mathrm{M}}2}.] \emph{Zero-mode normalization:}
    \begin{align}
        \bar{\check{M}}(\bm{0}) = \frac{1}{2d}. \label{checkM2}
    \end{align}
    \item[\criterion{\check{\mathrm{M}}3}.] \emph{Constant-modulus spectrum:}
    \begin{align}
        \abs{\bar{\check{M}}(\bm{m})} = \frac{1}{2d}, \quad \forall \bm{m}\in\mathcal{P}_{2d}. \label{checkM3}
    \end{align}
    \end{enumerate}
\end{proposition}
The key consequence for the present classification is ${\criterion{\check{\mathrm{M}}3}}$: it fixes the modulus of the projected Fourier stencil everywhere, so all remaining freedom is phase freedom.

\section{Constant-modulus phase normal form of valid stencils}
\label{sec_cons_mod_phase_norm_form}
This section expresses the projected-stencil validity criteria of \Cref{prop_fourier_domain_stencil_validity} in terms of a phase function.

By the constant-modulus condition \Cref{checkM3}, each value of the projected Fourier stencil has modulus ${1/(2d)}$ and hence admits a polar form. Thus there is a phase function ${\theta:\mathcal P_{2d}\to\reals/(2d\,\intsT)}$ such that
\begin{align}
    \bar{\check{M}}(\bm m)=\frac{1}{2d}\omegaT[2d][\theta(\bm m)],
    \qquad \bm m\in\mathcal P_{2d}.
    \label{eq_ch8_normal_form}
\end{align}
We call this the \emph{phase normal form}; the phase is defined modulo ${2d}$.

The phase constraints take the form of oddness relations modulo the relevant modulus. We therefore fix the terminology and negation notation here.

For ${N,T\in\intsT[>0]}$, a function ${f:\mathcal P_N\to\reals/(T\intsT)}$ is \emph{odd modulo ${T}$ on ${\mathcal P_N}$} when ${f(\bm n)+f(-\bm n)\equiv_T0}$ for every ${\bm n\in\mathcal P_N}$---we also call this \emph{untwisted oddness modulo ${T}$}.

\begin{notation}
\textbf{Notation.} For functions on ${\mathcal P_N}$, ${f(-\bm n)}$ means ${f(\{-\bm n\}_N)}$: negation of the argument is taken modulo the domain size ${N}$. For ${f:\mathcal P_N\to\reals/(T\intsT)}$, the minus sign in ${-f(\bm n)}$ instead negates the phase modulo ${T}$. The two moduli need not agree.
\end{notation}

The symplectic-Fourier criterion ${\criterion{\check{\mathrm{M}}1}}$ immediately constrains the phase function ${\theta}$ from \Cref{eq_ch8_normal_form}. Therefore, ${\criterion{\check{\mathrm{M}}1}}$, together with \Cref{eq_ch8_normal_form}, gives \({\omegaT[2d][-\theta(\bm{m})] = \omegaT[2d][\theta(-\bm{m})]}\),
and therefore ${\theta}$ satisfies the modular-oddness condition:
\begin{align}
    \criterion{\theta1}\quad &\tcboxmath[
        colback=white,
        colframe=black,
        boxrule=0.5pt,
        arc=2mm,
        left=2mm,
        right=2mm,
        top=1mm,
        bottom=1mm
    ]{
        \theta(\bm{m}) + \theta(-\bm{m})
        \equiv_{2d}
        0.
    }
    \label{eq_ch8_theta_odd}
\end{align}
Thus, ${\theta}$ is an odd function on ${\mathcal{P}_{2d}}$ modulo ${2d}$. Similarly, ${\criterion{\check{\mathrm{M}}2}}$ gives the origin condition:
\begin{align}
    \criterion{\theta2}\quad &\tcboxmath[
        colback=white,
        colframe=black,
        boxrule=0.5pt,
        arc=2mm,
        left=2mm,
        right=2mm,
        top=1mm,
        bottom=1mm
    ]{
        \theta(\bm{0})
        \equiv_{2d}
        0.
    }
    \label{eq_ch8_theta_origin}
\end{align}

A further constraint comes from the doubled-space projection. Since ${\bar{\check{M}}}$ lies in the image of ${\Proj}$, its values satisfy the projection-induced recurrence relation
\begin{align}
    \bar{\check{M}}(\bm{m}-d\,\bm{b})
    =
    \omegaT[2][\symprod{\bm{m}}{\bm{b}} - b_1 b_2 d]
    ~\bar{\check{M}}(\bm{m}),
    \label{eq_ch8_proj_recurrence_Mcheck}
\end{align}
for all ${\bm{m}\in\mathcal{P}_{2d}}$ and ${\bm{b}\in\intsT[2][2]}$, where ${\omegaT[2][a] = (-1)^a}$---see \Cref{def_exp_phase_notation}. Since ${\bm{m}-d\,\bm{b}\equiv_{2d}\bm{m}+d\,\bm{b}}$ for ${\bm{b}\in\intsT[2][2]}$, substituting \Cref{eq_ch8_normal_form} into \Cref{eq_ch8_proj_recurrence_Mcheck} yields the induced phase recurrence
\begin{align}
    \theta(\bm{m}+d\,\bm{b})
    \equiv_{2d}
    \theta(\bm{m})
    +
    d\,\{\symprod{\bm{m}}{\bm{b}} - b_1 b_2 \{d\}_2\}_2,
    \label{eq_ch8_theta_recurrence}
\end{align}
where ${\{\inputdot\}_T}$ denotes the modulo ${T}$ operation (with ${T\in\intsT[>0]}$), and ${\{d\}_2\in\{0,1\}}$ measures the parity of ${d}$, i.e., ${\{d_{\text{even}}\}_2 = 0}$ and ${\{d_{\text{odd}}\}_2 = 1}$. Hence, the phase function ${\theta}$ is not arbitrary on ${\mathcal{P}_{2d}}$. Once its values are known on a single ${d\times d}$ cell, the rest of the doubled phase space is fixed by \Cref{eq_ch8_theta_recurrence}---see also \Cref{fig_doubled_DWF_redundancy}.

Symplectic-Fourier validity therefore reduces to the oddness, origin and recurrence conditions \Cref{eq_ch8_theta_odd,eq_ch8_theta_origin,eq_ch8_theta_recurrence}. We now reduce these data to a single ${d\times d}$ base cell.

\section{Parameterization by base-cell phase functions}
\label{sec_param_by_base_cell_phase_funcs}
The recurrence relation \Cref{eq_ch8_theta_recurrence} links values separated by ${d\,\bm b}$ through a fixed parity-dependent rule, so a single ${d\times d}$ base cell suffices.

\subsection{Reduction to the base \titm{d\times d} cell}
\label{subsec_red_to_dxd_base_cell}
Every point ${\bm{m}\in\mathcal{P}_{2d}}$ may be written uniquely in the parameterized form
\begin{align}
    \bm{m}
    =
    \bm{\alpha}+d\,\bm{b},
    \qquad
    \bm{\alpha}\in\mathcal{P}_{d},
    \quad
    \bm{b}\in\intsT[2][2].
    \label{eq_ch8_m_base_cell_decomp}
\end{align}
Thus, ${\mathcal{P}_{2d}}$ decomposes into four translated copies of the base ${d\times d}$ cell, which serves as a fundamental domain for the shifts ${\bm{m}\mapsto\bm{m}+d\,\bm{b}}$.

For later use, we also record how negation acts in the parameterization ${\bm{m}=\bm{\alpha}+d\,\bm{b}}$. To track the carry produced by reducing the negation ${-\bm{\alpha}}$ back to the base cell, define
\begin{align}
    s(\alpha)
    \coloneqq
    1 - \DeltaT[d][\alpha]
    =
    \begin{cases}
        0, & \alpha = 0, \\
        1, & \alpha \neq 0,
    \end{cases}
    \qquad
    \bm{s}(\bm{\alpha})
    \coloneqq
    \vmat{s(\alpha_1)}{s(\alpha_2)}.
    \label{eq_ch8_s_function}
\end{align}
Then, if \({\bm{m}=\bm{\alpha}+d\,\bm{b}}\),
its negation modulo ${2d}$ may be written as
\begin{align}
    -\bm{m}
    &\equiv_{2d}
    \bm{\alpha}'+d\,\bm{b}',
    \nonumber \\
    \bm{\alpha}'&=\{-\bm{\alpha}\}_{d},\nonumber \\
    \bm{b}'&=\{\bm{b}+\bm{s}(\bm{\alpha})\}_{2}.
    \label{eq_ch8_negation_base_cell}
\end{align}
This is a carry rule, not simply the negation of the bit vector ${\bm b}$. If ${\alpha_i=0}$, then ${-d\,b_i\equiv_{2d}d\,b_i}$; otherwise, ${-\alpha_i-d\,b_i\equiv_{2d}(d-\alpha_i)+d\,\{b_i+1\}_2}$. Thus, each nonzero component produces exactly the carry recorded by ${s(\alpha_i)}$.

\subsection{Base-cell phase function \titm{\nu}}
Evaluating \Cref{eq_ch8_theta_recurrence} at ${\bm m=\bm\alpha}$ motivates isolating the base-cell phase function ${\nu}$ from which ${\theta}$ is reconstructed.

\begin{definition}[Base-cell phase function ${\nu}$]
\label{def_ch8_base_cell_phase_function}
    The \emph{base-cell phase function ${\nu}$} is the restriction of ${\theta}$ to the canonical ${d\times d}$ base cell of ${\mathcal{P}_{2d}}$, namely
    \begin{align}
    \label{eq_ch8_nu_def}
        \nu
        &:
        \mathcal{P}_{d} \to \reals/(2d\,\intsT), \nonumber \\
        \nu(\bm{\alpha})
        &\coloneqq
        \theta(\bm{\alpha}+d\,\bm{0}).
    \end{align}
\end{definition}

It is convenient to also define the fixed translation-dependent canonical term%
\footnote{Since ${d\equiv_2\{d\}_2}$, either may be used inside the remainder defining ${\beta_c}$; the latter makes the parity dependence explicit.}%
\begin{align}
    \theta_c(\bm{\alpha},\bm{b})
    &\coloneqq
    d\,\beta_c(\bm{\alpha},\bm{b})
    \in
    \{0,d\}
    \label{eq_ch8_theta_c_def} \\
    \text{where}
    \quad
    \beta_c(\bm{\alpha},\bm{b})
    &\coloneqq
    \{\symprod{\bm{\alpha}}{\bm{b}} - b_1 b_2 \{d\}_2\}_2 \in \{0,1\},
    \label{eq_ch8_beta_c_def}
\end{align}
where the subscript ${c}$ stands for canonical. Evaluating \Cref{eq_ch8_theta_recurrence} at ${\bm m=\bm\alpha}$ gives the base-cell reconstruction condition:
\begin{align}
    \criterion{\theta3}\quad &\tcboxmath[
        colback=white,
        colframe=black,
        boxrule=0.5pt,
        arc=2mm,
        left=2mm,
        right=2mm,
        top=1mm,
        bottom=1mm
    ]{
        \theta(\bm{\alpha}+d\,\bm{b})
        \equiv_{2d}
        \nu(\bm{\alpha})+\theta_c(\bm{\alpha},\bm{b}).
    }
    \label{eq_ch8_theta_nu_theta_c}
\end{align}
Thus, once ${\nu}$ is specified on the base cell, the function ${\theta}$ on ${\mathcal{P}_{2d}}$ is completely determined.

Here and below, we use \emph{admissible} for projected-Fourier and base-cell phase data that parameterize valid projected stencils.

We define an \emph{admissible doubled-cell phase function} to be a function ${\theta:\mathcal P_{2d}\to\reals/(2d\,\intsT)}$ satisfying modular oddness ${\criterion{\theta1}}$ [\Cref{eq_ch8_theta_odd}], the origin condition ${\criterion{\theta2}}$ [\Cref{eq_ch8_theta_origin}] and the base-cell recurrence ${\criterion{\theta3}}$ [\Cref{eq_ch8_theta_nu_theta_c}].

It remains to determine which base-cell phase functions ${\nu}$ give rise to admissible projected stencils. This is the subject of the next subsection.

\subsection{Constraints on the base-cell phase function \titm{\nu} for admissibility}
\label{subsec_constraints_on_base_cell_func_nu}
The origin condition passes directly to the base-cell phase function. By \Cref{def_ch8_base_cell_phase_function}, ${\nu(\bm{0}) = \theta(\bm{0} + d\,\bm{0}) = \theta(\bm{0})}$, and hence by the origin condition of ${\theta}$, we have
\begin{align}
    \tcboxmath[
        colback=white,
        colframe=black,
        boxrule=0.5pt,
        arc=2mm,
        left=2mm,
        right=2mm,
        top=1mm,
        bottom=1mm
    ]{
        \nu(\bm{0})
        \equiv_{2d}
        0.
    }
    \label{eq_ch8_nu_origin}
\end{align}
The subtlety in passing from the ordinary modular oddness of ${\theta}$ to a condition on its base-cell restriction ${\nu}$ is that the canonical ${d\times d}$ base cell, viewed inside ${\mathcal{P}_{2d}}$, is not generally preserved by negation modulo ${2d}$. For ${\bm{\alpha}\in\mathcal{P}_{d}}$, the negation of the canonical representative ${\bm{\alpha}+d\,\bm{0}}$ is represented by ${\{-\bm{\alpha}\}_{d}+d\,\bm{s}(\bm{\alpha})}$, rather than by the canonical representative ${\{-\bm{\alpha}\}_{d}+d\,\bm{0}}$---see \Cref{eq_ch8_negation_base_cell}. Consequently, expressing the value of ${\theta}$ at the negated doubled-lattice label in terms of ${\nu}$ introduces the recurrence phase associated with this carry. This is the source of the base-cell constraint derived below.

Taking the canonical representative ${\bm m=\bm\alpha+d\,\bm0}$, modular oddness, the negation rule \Cref{eq_ch8_negation_base_cell} and reconstruction \Cref{eq_ch8_theta_nu_theta_c} give
\begin{align}
 \nu(\bm\alpha)+\nu(\{-\bm\alpha\}_d)
 &\equiv_{2d}-\theta_c(\{-\bm\alpha\}_d,\bm s(\bm\alpha))\nonumber\\
 &\equiv_{2d}d\,\beta_c(\{-\bm\alpha\}_d,\bm s(\bm\alpha)).
 \label{eq_ch8_nu_pre_twisted_oddness}
\end{align}
The second congruence uses ${-\theta_c\equiv_{2d}\theta_c}$ because ${\theta_c\in\{0,d\}}$. To simplify the right-hand side, use ${\{-\alpha_i\}_d=s(\alpha_i)(d-\alpha_i)}$, which gives
\begin{align}
 \symprod{\{-\bm\alpha\}_d}{\bm s(\bm\alpha)}
 =s(\alpha_1)s(\alpha_2)(\alpha_2-\alpha_1).
\end{align}
Thus only the off-axis case contributes. Since ${\alpha_2-\alpha_1\equiv_2\alpha_1+\alpha_2}$, substitution into \Cref{eq_ch8_beta_c_def} gives the canonical twist below.

Hence, ${\nu}$ satisfies the following \emph{twisted oddness relation modulo ${2d}$}, which is the base-cell analog of the modular oddness of ${\theta}$:
\begin{align}
    \tcboxmath[
        colback=white,
        colframe=black,
        boxrule=0.5pt,
        arc=2mm,
        left=2mm,
        right=2mm,
        top=1mm,
        bottom=1mm
    ]{
        \nu(\bm{\alpha})+\nu(\{-\bm{\alpha}\}_{d})
        \equiv_{2d}
        \Xi(\bm{\alpha}),
    }
    \label{eq_ch8_twisted_oddness}
\end{align}
where
\begin{align}
    \Xi(\bm{\alpha})
    \coloneqq\,&
    d\,s(\alpha_1)s(\alpha_2)\{\alpha_1+\alpha_2+\{d\}_2\}_2 \\
    =&
    \begin{cases}
        d\,\{\alpha_1+\alpha_2+\{d\}_2\}_2,
        & \alpha_1\not\equiv_d 0\ \text{and}\ \alpha_2\not\equiv_d 0,
        \\
        0,
        & \alpha_1\equiv_d 0\ \text{or}\ \alpha_2\equiv_d 0,
    \end{cases}
    \label{eq_ch8_Xi_def}
\end{align}
and so ${\Xi(\bm{\alpha}) \in \{0,d\}}$.

\begin{definition}[Twisted oddness modulo ${2d}$ on ${\mathcal{P}_{d}}$]
\label{def_twisted_oddness}
    Let
    ${f:\mathcal{P}_{d}\to\reals/(2d\,\intsT).}$
    Here, the argument involution is ${\bm{\alpha}\mapsto\{-\bm{\alpha}\}_{d}}$ on ${\mathcal{P}_{d}}$, while congruences between function values are taken modulo ${2d}$. Then ${f}$ is said to satisfy \emph{twisted oddness} with twist ${\Xi:\mathcal{P}_{d}\to\reals/(2d\,\intsT)}$ if
    \begin{align}
        f(\bm{\alpha})
        +
        f(\{-\bm{\alpha}\}_{d})
        \equiv_{2d}
        \Xi(\bm{\alpha}),
        \qquad
        \forall~\bm{\alpha}\in\mathcal{P}_{d}.
    \end{align}
\end{definition}
The twist ${\Xi(\bm{\alpha})}$ arises from the canonical phase term in the doubled-cell reconstruction when negation is expressed through base-cell representatives. Thus, ${\nu}$ is twisted odd modulo ${2d}$ on ${\mathcal P_d}$ with twist ${\Xi}$. If ${\Xi}$ were zero, this would be untwisted oddness with domain modulus ${d}$ and codomain modulus ${2d}$; the present canonical twist is generally nonzero. Its base-cell values for ${d=4}$ and ${d=5}$ are illustrated in \Cref{fig_ch8_Xi_support}.

\begin{figure}[!t]
    \centering
    \includegraphics[width=\columnwidth]{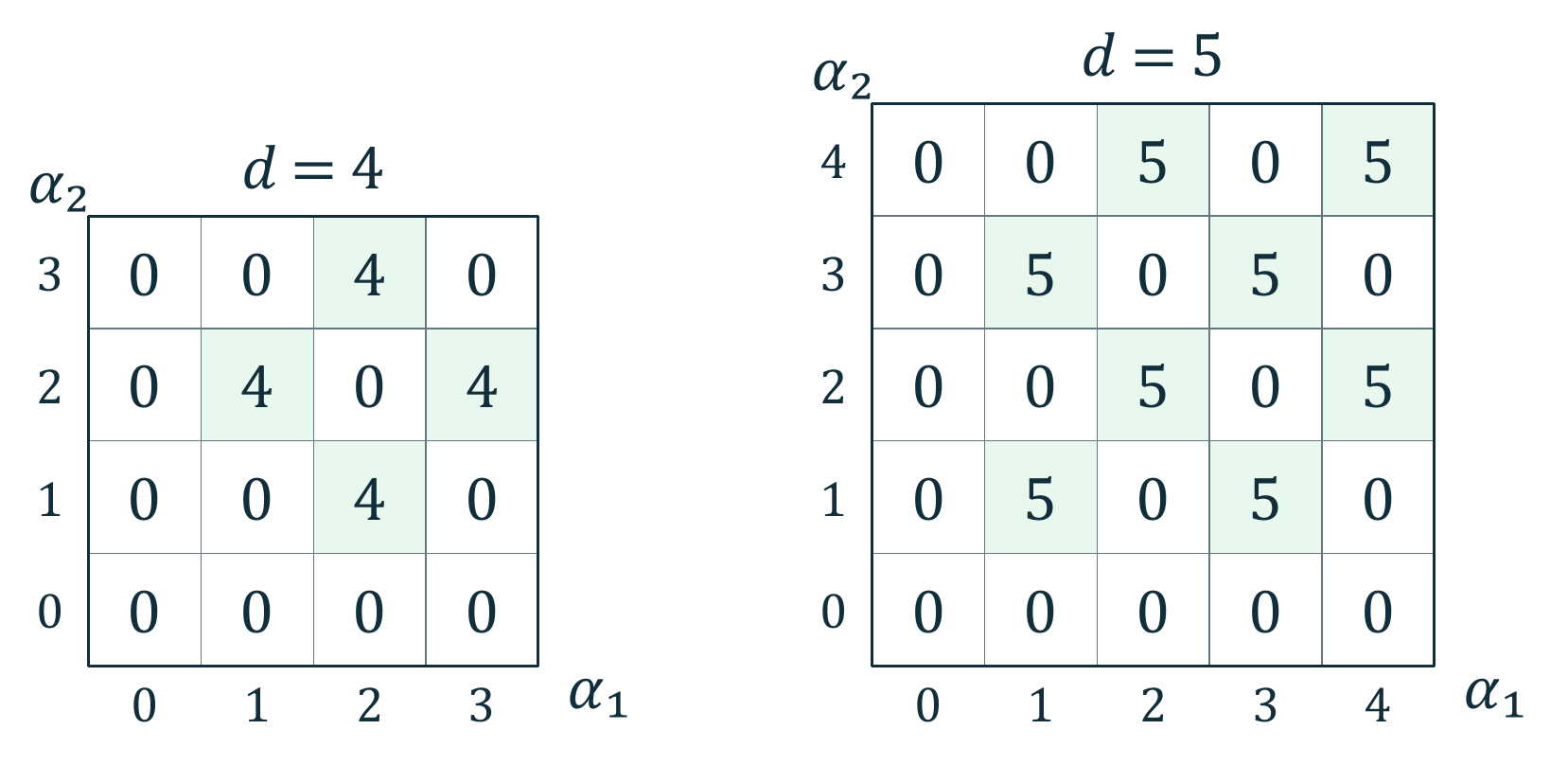}
    \caption[Canonical twist for even and odd dimensions]{Base-cell values of the canonical twist ${\Xi(\bm\alpha)}$ for ${d=4}$ (left) and ${d=5}$ (right), calculated from \Cref{eq_ch8_Xi_def}. Green cells have ${\Xi=d}$, shown explicitly as 4 or 5; unshaded cells have ${\Xi=0}$. The twist vanishes on both coordinate axes. Off-axis, ${\Xi(\bm\alpha)=d}$ exactly when the components have opposite parity for even ${d}$, or the same parity for odd ${d}$. Each pattern is invariant under ${\bm\alpha\mapsto\{-\bm\alpha\}_d}$, as required by twisted oddness.}
    \label{fig_ch8_Xi_support}
\end{figure}

\begin{remark}[Symmetry of the twist function]
    Evaluating the twisted-oddness relation at ${\{-\bm{\alpha}\}_{d}}$ gives
    ${f(\{-\bm\alpha\}_d)+f(\bm\alpha)\equiv_{2d}\Xi(\{-\bm\alpha\}_d)}$,
    since
    ${\{-\{-\bm{\alpha}\}_{d}\}_{d}=\bm{\alpha}}$.
    Comparing this with the original twisted-oddness relation therefore implies
    \begin{align}
        \Xi(\{-\bm{\alpha}\}_{d})
        \equiv_{2d}
        \Xi(\bm{\alpha}).
    \end{align}
    The canonical twist function ${\Xi}$ of \Cref{eq_ch8_Xi_def} satisfies this symmetry.
\end{remark}

For ${d\geq3}$, the base-cell phase function ${\nu}$ is therefore not an ordinary odd function on ${\mathcal{P}_{d}}$, but rather an odd function twisted by the parity-dependent canonical term ${\Xi(\bm{\alpha})}$.%
\footnote{For ${d=1,2}$, the twist function ${\Xi}$ vanishes, so twisted oddness reduces to ordinary oddness modulo ${2d}$.} %
This twisted oddness is the key structural constraint on the base-cell phase function ${\nu}$. We now collect the resulting base-cell constraints into a definition.

\begin{definition}[Admissible base-cell phase functions and admissible parameter set]
\label{def_ch8_nu_admissibility}
    Let
    ${\nu : \mathcal{P}_{d} \to \reals/(2d\,\intsT)}$
    be a \emph{base-cell phase function}, ${\theta : \mathcal{P}_{2d} \to \reals/(2d\,\intsT)}$ be its doubled-cell extension under \Cref{eq_ch8_theta_nu_theta_c}, and define the associated projected Fourier stencil by
    ${\bar{\check{M}}(\bm{m}) \coloneqq \frac{1}{2d}\omegaT[2d][\theta(\bm{m})],}$
    where ${\bm{m}\in\mathcal{P}_{2d}}$. An argument belonging to ${\mathcal{P}_{d}}$ is understood as its canonical representative in the base ${d\times d}$ cell of ${\mathcal{P}_{2d}}$. The function ${\nu}$ is called \emph{admissible} if it satisfies the following criteria:
    \begin{enumerate}[align=left]
    \item[\criterion{\nu1}.] \emph{Twisted oddness:}
    \begin{align}
        \nu(\bm{\alpha}) + \nu(\{-\bm{\alpha}\}_{d}) \equiv_{2d} \Xi(\bm{\alpha}), \quad \forall \bm{\alpha}\in\mathcal{P}_{d}.
    \end{align}
    \item[\criterion{\nu2}.] \emph{Origin condition:}
    \begin{align}
        \nu(\bm{0}) \equiv_{2d} 0.
    \end{align}
    \end{enumerate}
    The set of all admissible base-cell phase functions is the \emph{admissible parameter set}, denoted by
    \begin{align}
        \solspace
        \coloneqq
        \left\{\begin{aligned}
            &\nu:\mathcal{P}_{d}\to\reals/(2d\,\intsT)\\
            &\quad\mid\nu \text{ satisfies } \criterion{\nu1} \text{ and } \criterion{\nu2}
        \end{aligned}\right\}.
    \end{align}
\end{definition}

The preceding calculation shows how these criteria arise from the symplectic-Fourier validity conditions. The following proposition shows that they are not merely necessary but also sufficient.

\begin{proposition}[Base-cell form of stencil admissibility]
\label{prop_ch8_base_cell_form_of_stencil_admissibility}
    Under the doubled-cell projection recurrence, the constant-modulus normal form enforces ${\criterion{\check{\mathrm{M}}3}}$, while ${\criterion{\check{\mathrm{M}}1}}$ and ${\criterion{\check{\mathrm{M}}2}}$ are equivalent to ${\criterion{\nu1}}$ and ${\criterion{\nu2}}$, respectively.

    Equivalently, any base-cell phase function ${\nu}$ satisfying ${\criterion{\nu1}}$--${\criterion{\nu2}}$, extended to ${\theta}$ by \Cref{eq_ch8_theta_nu_theta_c} and inserted into the constant-modulus normal form, defines a projected Fourier stencil satisfying ${\criterion{\check{\mathrm{M}}1}}$--${\criterion{\check{\mathrm{M}}3}}$.
\end{proposition}

\begin{proof}
    Criterion ${\criterion{\check{\mathrm{M}}3}}$ is the constant-modulus condition \({\abs{\bar{\check{M}}(\bm{m})} = \frac{1}{2d}}\),
    and hence permits the phase normal form \({\bar{\check{M}}(\bm{m}) = \frac{1}{2d} \omegaT[2d][\theta(\bm{m})]}\).
    This condition is built into the reconstruction of ${\bar{\check{M}}}$ in \Cref{def_ch8_nu_admissibility}.

    Criterion ${\criterion{\check{\mathrm{M}}2}}$ fixes the zero-mode: \({\bar{\check{M}}(\bm{0}) = \frac{1}{2d}}\).
    In the phase normal form, this is equivalent to \({\nu(\bm{0}) \equiv_{2d} 0}\),
    which is ${\criterion{\nu2}}$.

    It remains to consider ${\criterion{\check{\mathrm{M}}1}}$. In the phase normal form, the Hermitian-symmetry condition \({\bar{\check{M}}(\bm{m})^{*} = \bar{\check{M}}(-\bm{m})}\)
    is equivalent to \({\theta(-\bm{m}) \equiv_{2d} -\theta(\bm{m})}\).
    As derived in \Cref{eq_ch8_nu_pre_twisted_oddness,eq_ch8_twisted_oddness}, restricting this relation to the base cell gives \({\nu(\bm{\alpha}) + \nu(\{-\bm{\alpha}\}_{d}) \equiv_{2d} \Xi(\bm{\alpha})}\),
    which is ${\criterion{\nu1}}$.

    Conversely, suppose that ${\criterion{\nu1}}$ holds, and extend ${\nu}$ to ${\theta}$ using \Cref{eq_ch8_theta_nu_theta_c}. The extension formula fixes ${\theta}$ relative to each base-cell representative, but projection compatibility requires the same recurrence to hold under a ${d}$-shift from an arbitrary point of ${\mathcal{P}_{2d}}$. This check depends only on the extension formula.

    Write ${\bm{m}'=\bm{\alpha}'+d\,\bm{c}}$, where ${\bm{\alpha}'\in\mathcal{P}_{d}}$ and ${\bm{c}\in\intsT[2][2]}$. For ${\bm{r}\in\intsT[2][2]}$, set ${\tilde{\bm{c}}\coloneqq\{\bm{c}+\bm{r}\}_{2}}$. Therefore, \({\bm{m}'+d\,\bm{r} \equiv_{2d} \bm{\alpha}'+d\,\tilde{\bm{c}}}\).
    Since ${\bm m'}$ and ${\bm m'+d\,\bm r}$ have the same base-cell representative ${\bm\alpha'}$, their ${\nu(\bm\alpha')}$ contributions cancel.
    It therefore suffices to compute the change in the canonical term under ${\bm{c}\mapsto \tilde{\bm{c}}}$. Applying \Cref{eq_ch8_beta_c_def} to both occurrences and taking their difference gives
    \begin{align}
        &\beta_c(\bm{\alpha}',\tilde{\bm{c}})-\beta_c(\bm{\alpha}',\bm{c})\nonumber\\
        &\quad\equiv_{2}
        \symprod{\bm{\alpha}'}{\bm{r}}
        -
        \{d\}_{2}
        \left(
            c_1r_2+r_1c_2+r_1r_2
        \right)
        \nonumber\\
        &\quad\equiv_{2}
        \symprod{\bm{m}'}{\bm{r}}
        -
        r_1r_2\{d\}_{2},
    \end{align}
    where the second congruence uses ${\bm{m}'=\bm{\alpha}'+d\,\bm{c}}$ and ${d\equiv_{2}\{d\}_{2}}$. Multiplication by ${d}$ lifts a congruence modulo ${2}$ to one modulo ${2d}$. Hence,
    \begin{align}
        \theta(\bm{m}'+d\,\bm{r}) - \theta(\bm{m}')
        &\equiv_{2d}
        d\,\left\{
            \symprod{\bm{m}'}{\bm{r}}
            -
            r_1r_2\{d\}_{2}
        \right\}_{2}.
    \end{align}
    This is exactly \Cref{eq_ch8_theta_recurrence}. Substitution into the constant-modulus normal form therefore shows that the reconstructed ${\bar{\check{M}}}$ satisfies \Cref{eq_ch8_proj_recurrence_Mcheck} and hence ${\bar{\check{M}}\in\im(\Proj)}$, since substituting this recurrence into the explicit projector formula \Cref{eq_proj_explicit} gives ${\Proj\bar{\check{M}}=\bar{\check{M}}}$.

    It remains to verify the Hermitian symmetry of the reconstructed projected Fourier stencil ${\bar{\check{M}}}$ via ${\criterion{\nu1}}$. For ${\bm{m}=\bm{\alpha}+d\,\bm{b}}$, set
    \begin{align}
        \bm{\alpha}^{-}
        \coloneqq
        \{-\bm{\alpha}\}_{d},
        \qquad
        \bm{b}^{-}
        \coloneqq
        \{\bm{b}+\bm{s}(\bm{\alpha})\}_{2}.
    \end{align}
    By \Cref{eq_ch8_negation_base_cell}, \({-\bm{m} \equiv_{2d} \bm{\alpha}^{-}+d\,\bm{b}^{-}}\).
    Therefore,
    \begin{align}
        \theta(\bm{m})+\theta(-\bm{m})
        &\equiv_{2d}
        \nu(\bm{\alpha})
        +
        \nu(\bm{\alpha}^{-})
        +
        d\,\beta_c(\bm{\alpha},\bm{b}) \nonumber \\
        &\qquad 
        +
        d\,\beta_c(\bm{\alpha}^{-},\bm{b}^{-})
        \nonumber\\
        &\equiv_{2d}
        \Xi(\bm{\alpha})
        +
        d\,\beta_c(\bm{\alpha},\bm{b})
        +
        d\,\beta_c(\bm{\alpha}^{-},\bm{b}^{-}),
        \label{eq_ch8_converse_oddness_intermediate}
    \end{align}
    where the second congruence uses ${\criterion{\nu1}}$. Moreover, direct substitution into \Cref{eq_ch8_beta_c_def,eq_ch8_Xi_def}, using ${\bm{\alpha}^{-}=\{-\bm{\alpha}\}_{d}}$, gives
    ${\Xi(\bm\alpha)\equiv_{2d}d\,\beta_c(\bm\alpha^-,\bm s(\bm\alpha))}$.

    Then, direct substitution into \Cref{eq_ch8_beta_c_def}, together with ${\bm{\alpha}+\bm{\alpha}^{-}=d\,\bm{s}(\bm{\alpha})}$ and ${\bm{b}^{-}\equiv_2\bm{b}+\bm{s}(\bm{\alpha})}$, gives
    \begin{align}
        &\beta_c(\bm{\alpha},\bm{b})
        +
        \beta_c(\bm{\alpha}^{-},\bm{b}^{-})
        +
        \beta_c(\bm{\alpha}^{-},\bm{s}(\bm{\alpha}))
        \nonumber\\
        &\qquad\equiv_2
        \{d\}_{2}
        \left(
            \symprod{\bm{s}(\bm{\alpha})}{\bm{b}}
            -
            b_1 s(\alpha_2)
            -
            s(\alpha_1) b_2
        \right)
        \nonumber\\
        &\qquad\equiv_2
        -2\{d\}_{2} s(\alpha_2) b_1
        \equiv_2
        0.
        \label{eq_ch8_converse_carry_identity}
    \end{align}
    Substitution into
    \Cref{eq_ch8_converse_oddness_intermediate} yields \({\theta(\bm{m})+\theta(-\bm{m}) \equiv_{2d} 0}\).
    Hence, the reconstructed phase function satisfies \({\theta(-\bm{m}) \equiv_{2d} -\theta(\bm{m}), \qquad \forall~\bm{m}\in\mathcal{P}_{2d}}\),
    and the reconstructed projected Fourier stencil satisfies ${\criterion{\check{\mathrm{M}}1}}$.

    Thus, ${\criterion{\check{\mathrm{M}}3}}$ is built into the phase normal form, while ${\criterion{\check{\mathrm{M}}2}}$ and ${\criterion{\check{\mathrm{M}}1}}$ are respectively equivalent to ${\criterion{\nu2}}$ and ${\criterion{\nu1}}$. This proves the claimed equivalence.
\end{proof}

There is no independent ${\criterion{\nu3}}$ criterion: ${\criterion{\check{\mathrm{M}}3}}$ is built into the constant-modulus reconstruction. Under the reconstruction in \Cref{eq_ch8_theta_nu_theta_c,eq_ch8_normal_form}, ${\criterion{\check{\mathrm{M}}1}\Longleftrightarrow\criterion{\nu1}}$ and ${\criterion{\check{\mathrm{M}}2}\Longleftrightarrow\criterion{\nu2}}$. Thus, the three descriptions of admissibility are equivalent:
\begin{align}
        \tcboxmath[
        colback=white,
        colframe=black,
        boxrule=0.5pt,
        arc=2mm,
        left=2mm,
        right=2mm,
        top=1mm,
        bottom=1mm
    ]{
    \criterion{\theta1}\text{--}\criterion{\theta3}
    \quad\Longleftrightarrow\quad
    \criterion{\nu1}\text{--}\criterion{\nu2}
    \quad\Longleftrightarrow\quad
    \criterion{\check{\mathrm{M}}1}\text{--}\criterion{\check{\mathrm{M}}3}.
    }
    \label{eq_ch8_theta_nu_M_admissibility_chain}
\end{align}
Thus, the doubled-cell conditions ${\criterion{\theta1}}$--${\criterion{\theta3}}$ are necessary and sufficient for the reconstructed stencil to be valid.

Let ${\bar M_\nu}$ be the projected stencil reconstructed from ${\nu\in\solspace}$ by \Cref{eq_ch8_theta_nu_theta_c,eq_ch8_normal_form} and the inverse SDFT. Projected-stencil uniqueness then gives the full family of valid unprojected stencils:
\begin{align}
    M=\bar M_\nu+K,\qquad \nu\in\solspace,\quad K\in\ker\Proj.
    \label{eq_all_unprojected_stencils}
\end{align}
The phase function uniquely fixes the projected stencil and labeled representation; ${K}$ is arbitrary redundant freedom in the kernel of the projection.

Of these two base-cell conditions, twisted oddness deserves particular emphasis because it is the base-cell form of the symplectic-Fourier Hermitian-symmetry condition ${\bar{\check{M}}(\bm{m})^{*}=\bar{\check{M}}(-\bm{m})}$. Thus, twisted oddness is not an additional phase convention---it expresses, on the ${d\times d}$ base cell, the Hermiticity requirement that gives the induced ${M}$-PPOs property \criterion{\mathrm{A}1}.

The classification problem is now one of choosing an admissible base-cell phase function ${\nu\in\solspace}$. To determine the topology and remaining freedom of ${\solspace}$, consider the involution ${\bm{\alpha}\mapsto\{-\bm{\alpha}\}_{d}}$ on the base cell: non-fixed points (${\bm{\alpha}\neq \{-\bm{\alpha}\}_{d}}$) occur in pairs and give continuous phase choices, whereas fixed points (${\bm{\alpha}=\{-\bm{\alpha}\}_{d}}$), or \emph{self-negative points}, impose self-consistency conditions ${2\nu(\bm{\alpha})\equiv_{2d}\Xi(\bm{\alpha})}$ through ${\criterion{\nu1}}$; after the origin is fixed by ${\criterion{\nu2}}$, any remaining fixed points give discrete choices. This is explored next.

\section{Topology and marginally constrained subspaces}
\label{sec_topo_and_margin_constr_subspaces}
We now determine the topology of the admissible parameter set ${\solspace}$ of \Cref{def_ch8_nu_admissibility}.

\begin{definition}[Admissible parameter set, topology and admissible parameter space]
\label{def_ch8_admissible_parameter_space_topology}
    The ambient parameter space is the finite Cartesian product
    ${\prod_{\bm{\alpha}\in\mathcal{P}_{d}} \reals/(2d\,\intsT),}$
    identified with the set of all base-cell phase functions ${\nu:\mathcal{P}_{d}\to\reals/(2d\,\intsT)}$.

    Each factor ${\reals/(2d\,\intsT)}$ is equipped with the quotient topology induced by the standard topology on ${\reals}$, and the finite Cartesian product is equipped with the corresponding product topology.

    The \emph{admissible parameter set} is the subset ${\solspace}$ of this ambient product space cut out by the admissibility criteria ${\criterion{\nu1}}$ and ${\criterion{\nu2}}$:
    ${\solspace \subseteq \prod_{\bm{\alpha}\in\mathcal{P}_{d}} \reals/(2d\,\intsT).}$
    The \emph{topology} ${\topol}$ on ${\solspace}$ is the subspace topology inherited from this ambient product topology.

    The topological space defined by the pair ${(\solspace,\topol)}$ is called the \emph{admissible parameter space}.
\end{definition}

The ambient product is an indexed Cartesian product containing one copy of the phase circle ${\reals/(2d\,\intsT)\cong S^1}$ for each of the ${d^2}$ base-cell points, and is therefore homeomorphic to ${(S^1)^{d^2}}$. Its elements are the tuples ${(\nu(\bm{\alpha}))_{\bm{\alpha}\in\mathcal{P}_{d}}}$ of phase values assigned on the base cell. The open sets of ${\solspace}$ are precisely its intersections with open sets of the ambient product.

\subsection{Continuous and discrete degrees of freedom}
\label{subsec_cont_and_disc_degree_freedom}
The parameterization above organizes the admissible phase data according to distinct negation pairs and self-negative points under the involution \({\bm{\alpha} \mapsto \{-\bm{\alpha}\}_d}\).
For points that are not fixed by this involution, \Cref{eq_ch8_twisted_oddness} shows that choosing ${\nu(\bm{\alpha})}$ determines ${\nu(\{-\bm{\alpha}\}_d)}$ uniquely. Since ${\nu}$ takes values in ${\reals/2d\,\intsT}$, each such free choice is a continuous phase parameter---that is, an ${S^1}$-valued degree of freedom.

At a fixed point of this involution, the twisted-oddness relation becomes the self-consistency condition:
\begin{align}
    2 \nu(\bm{\alpha}) \equiv_{2d} \Xi(\bm{\alpha}).
    \label{eq_ch8_self_negative_constraint_intro}
\end{align}
At these self-negative points, ${\Xi(\bm{\alpha})=0}$: there is only the origin for odd ${d}$, and four points for even ${d}$:
\begin{align}
    (0,0),
    \qquad
    \left(\frac{d}{2},0\right),
    \qquad
    \left(0,\frac{d}{2}\right),
    \qquad
    \left(\frac{d}{2},\frac{d}{2}\right).
    \label{eq_ch8_even_self_negative_points}
\end{align}
Thus \Cref{eq_ch8_self_negative_constraint_intro} permits two values, separated by ${d}$. Imposing ${\nu(\bm0)=0}$ by ${\criterion{\nu2}}$ removes the origin choice; each remaining self-negative point contributes a discrete ${\intsT[2]}$ choice rather than a circle---see \Cref{fig_ch8_low_dim_examples} below.

\subsection{Parameter-space structure and low-dimensional examples}
\label{subsec_topo_theo_and_low_dim_examps}
The following theorem counts the independent phase-pair and self-negative-point choices and records the topological type of ${\solspace}$.

\begin{theorem}[Structure of the admissible parameter space]
\label{thm_ch8_topology_of_admissible_stencils}
    Let ${\solspace}$ be the admissible parameter set of \Cref{def_ch8_nu_admissibility}. Then ${\solspace}$ is homeomorphic to
    \begin{align}
        \solspace
        \cong
        \left(S^1\right)^{N_c(d)}
        \times
        \intsT[2][N_b(d)],
    \end{align}
    where each ${\intsT[2]}$ factor is equipped with the discrete topology. Here,
    \begin{align}
        N_c(d)\coloneqq\frac{d^2-F(d)}{2},
        \quad\text{and}\quad N_b(d)\coloneqq F(d)-1
    \end{align}
    count the independent circle-valued parameters and binary (bit) parameters, respectively. The quantity
    \begin{align}
        F(d)
        \coloneqq
        \# 
        \left\{
        \bm{\alpha}\in\mathcal{P}_{d}
        :
        \bm{\alpha}=\{-\bm{\alpha}\}_{d} 
        \right\}
        =
        \begin{cases}
            1, & d \text{ odd}, \\
            4, & d \text{ even},
        \end{cases}
    \end{align}
    is the number of negation-fixed points in ${\mathcal{P}_{d}}$. Equivalently,
    \begin{align}
    \label{eq_ch8_topology_solution_space}
        \solspace
        \cong
        \begin{cases}
            \left(S^1\right)^{\frac{d^2-1}{2}},
            & d \text{ odd},
            \\[1mm]
            \left(S^1\right)^{\frac{d^2-4}{2}}
            \times
            \intsT[2][3],
            & d \text{ even}.
        \end{cases}
    \end{align}
\end{theorem}

\begin{proof}
The involution ${\bm{\alpha}\mapsto\{-\bm{\alpha}\}_d}$ partitions the base cell into distinct pairs and self-negative points. On a distinct pair, freely choosing ${\nu(\bm{\alpha})\in\reals/(2d\,\intsT)}$ determines its partner by
\begin{align}
 \nu(\{-\bm{\alpha}\}_d)\equiv_{2d}\Xi(\bm{\alpha})-\nu(\bm{\alpha}).
\end{align}
The relation is consistent in both directions because ${\Xi(\{-\bm{\alpha}\}_d)=\Xi(\bm{\alpha})}$. Each pair therefore contributes one circle-valued degree of freedom.

At a self-negative point, ${2\alpha_i\equiv_d0}$ for both components. There is only the origin for odd ${d}$, while for even ${d}$ each coordinate is ${0}$ or ${d/2}$, giving four such points. Thus ${F(d)=1}$ or ${4}$, respectively. At each of these points ${\Xi=0}$, and twisted oddness gives
\begin{align}
 2\nu(\bm{\alpha})\equiv_{2d}0,
 \label{eq_ch8_self_negative_constraint_proof}
\end{align}
with precisely the two choices ${0,d}$. The origin condition fixes the origin choice to zero. Consequently, ${N_b(d)=F(d)-1}$ and ${N_c(d)=(d^2-F(d))/2}$.

Choose one representative from each distinct pair and identify each non-origin self-negative choice ${0,d}$ with a bit. Restriction to these independent data gives a bijection
\begin{align}
 \solspace\longrightarrow(\reals/(2d\,\intsT))^{N_c(d)}\times\intsT[2][N_b(d)].
\end{align}
Its inverse sets ${\nu(\bm{0})=0}$, assigns the chosen circle and bit values and reconstructs every partner by twisted oddness. Thus every base-cell point receives exactly one value, and the reconstructed function satisfies ${\criterion{\nu1}}$--${\criterion{\nu2}}$. Restriction is continuous in the topology of \Cref{def_ch8_admissible_parameter_space_topology}; reconstruction is continuous because the partner varies by the circle map ${x\mapsto\Xi(\bm{\alpha})-x}$, while the bit choices carry the discrete topology. The bijection is therefore a homeomorphism, proving \Cref{thm_ch8_topology_of_admissible_stencils}, including the trivial singleton for ${d=1}$.
\end{proof}

\begin{remark}[Equivalent form without ${F(d)}$]
\label{remark_equivalent_parameter_count_no_F}
    Equivalently, one may write
    ${N_c(d)=\frac{d^2 - 1 - 3\{d-1\}_2}{2}}$, and ${N_b(d)=3\{d-1\}_2}$.
\end{remark}

Thus, the odd-dimensional admissible space is a connected torus, while the three independent even-dimensional bits distinguish eight disjoint copies of ${(S^1)^{(d^2-4)/2}}$.

For odd ${d}$, only the origin is fixed by negation, whereas for even ${d}$, there are four self-negative points, given in \Cref{eq_ch8_even_self_negative_points}. For example, at ${d=3}$ the eight non-origin points form four negation pairs and contribute four circles. At ${d=4}$, fixing the origin leaves three bits and six circles. At ${d=1}$, there are no free parameters and only the trivial solution. The preceding proof gives the componentwise count.

The first few cases are listed in \Cref{tab_ch8_low_dim_topology_counts}. This table records the topological type of the unconstrained admissible parameter space obtained solely from validity, before imposing any additional marginal or displacement-operator algebraic restrictions.

\begin{table}[tbp]
    \centering
    \small
    \renewcommand{\arraystretch}{1.25}
    \begin{tabular*}{\linewidth}{@{\extracolsep{\fill}}ccc@{}}
        \hline\hline
        ${d}$ & Free parameters & ${\solspace}$ \\
        \hline
        \noalign{\vskip 1pt}
        Odd ${d}$
        &
        ${\dfrac{d^2-1}{2}}$ circles
        &
        ${\left(S^1\right)^{\frac{d^2-1}{2}}}$ \\
        
        Even ${d}$
        &
        ${\dfrac{d^2-4}{2}}$ circles ${+3}$ bits
        &
        ${\left(S^1\right)^{\frac{d^2-4}{2}}\times\intsT[2][3]}$ \\

        ${1}$ & None & ${\{\ast\}}$ \\

        ${2}$ & ${3}$ bits & ${\intsT[2][3]}$ \\
        
        ${3}$ & ${4}$ circles & ${(S^1)^4}$ \\
        
        ${4}$ & ${6}$ circles ${+3}$ bits & ${(S^1)^6\times\intsT[2][3]}$ \\
        
        ${5}$ & ${12}$ circles & ${(S^1)^{12}}$ \\
        
        ${6}$ & ${16}$ circles ${+3}$ bits & ${(S^1)^{16}\times\intsT[2][3]}$ \\
        
        ${7}$ & ${24}$ circles & ${(S^1)^{24}}$ \\
        \hline\hline
    \end{tabular*}
    \caption[Admissible parameter spaces by dimensions]{Topological type of the admissible parameter space ${\solspace}$ for general odd and even dimensions and for the first few nontrivial dimensions. The table records ${\solspace}$ obtained from the ${\nu}$-admissibility criteria alone. Continuous phase-pair choices contribute ${S^1}$ circle factors, while self-negative points contribute ${\intsT[2]}$ bit factors after the origin is fixed. The trivial case ${d=1}$ is a singleton.}
    \label{tab_ch8_low_dim_topology_counts}
\end{table}

These low-dimensional parameter counts are further illustrated schematically in \Cref{fig_ch8_low_dim_examples}, where gray shadings denote continuous free ${S^1}$-type parameters and green shadings denote discrete free ${\intsT[2]}$-type choices.

\begin{figure}[h]
    \centering
    \includegraphics[width=\columnwidth]{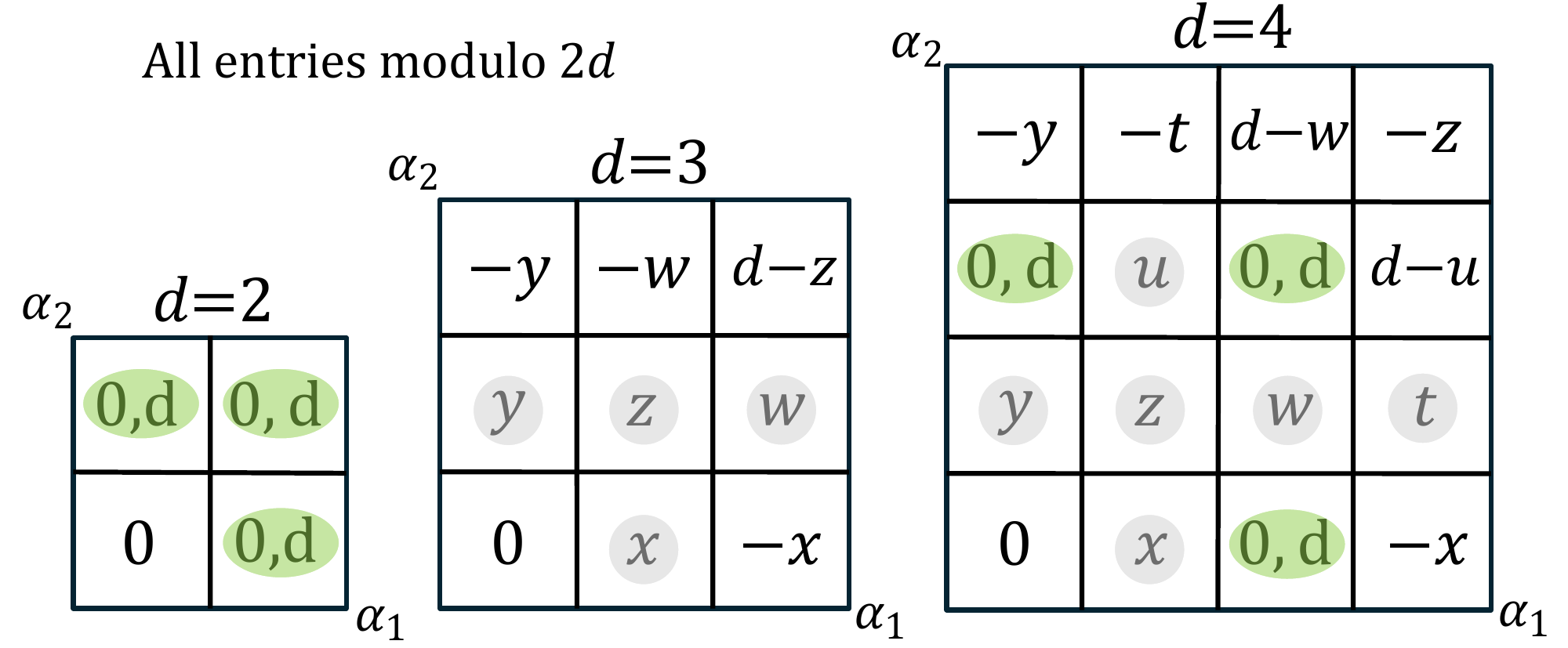}
    \caption[Free parameter structure for ${d=2,3,4}$]{Low-dimensional examples of the free parameter structure determined by twisted oddness on the base ${d\times d}$ cell. Left: ${d=2}$, where the parameter space is ${\intsT[2][3]}$. Middle: ${d=3}$, where the parameter space is ${(S^1)^4}$. Right: ${d=4}$, where the parameter space is ${(S^1)^6\times \intsT[2][3]}$. Gray shadings denote independent ${S^1}$-type parameters of ${\nu}$. Green shadings denote discrete ${\intsT[2]}$-type choices at self-negative points. The independent continuous choices are labeled by ${x, y, z, w, u, t}$. The displayed partners satisfy ${\nu(\{-\bm{\alpha}\}_d)\equiv_{2d}\Xi(\bm{\alpha})-\nu(\bm{\alpha})}$, with ${\Xi}$ given by \Cref{eq_ch8_Xi_def}; all entries are read modulo ${2d}$.}
    \label{fig_ch8_low_dim_examples}
\end{figure}

The topology above follows solely from validity. Additional physical or algebraic requirements can remove phase degrees of freedom, discretize continuous choices or force nonexistence. Tracking these changes in ${\solspace}$ distinguishes constraints selecting valid DWF subfamilies from those incompatible with validity itself.

\begin{notation}
\textbf{Base-cell label convention.} Whenever ${\bm{\alpha}\in\mathcal{P}_{d}}$ appears as the argument of a doubled-lattice quantity such as ${\hat{V}^{(2d)}}$ or ${\bar{\check{M}}}$, it denotes the ${\bm{b}=\bm{0}}$ representative ${\bm{\alpha}+d\,\bm{0}\in\mathcal{P}_{2d}}$. This convention suppresses only the explicit ${+d\,\bm{0}}$; carry terms arising when sums or negations are reduced to the base cell are retained.
\end{notation}

\subsection{Stencil-induced WHDOs}
\label{subsec_stencil_ind_M_WHDOs}
Given the ${M}$-induced PPO frame ${\{\hat{A}^{M}(\bm{\alpha})\}_{\bm{\alpha}\in\mathcal{P}_{d}}}$ associated with a valid stencil ${M}$, it is natural to introduce the corresponding displacement-operator family by taking its ${d}$-sized SDFT [see \Cref{def_SDFT}]. This is the direct analog of the symplectic-Fourier relation between phase-point operators and WHDOs in the standard discrete Wigner formalism.

\begin{definition}[${M}$-WHDOs]
\label{def_M_WHDOs}
    Let ${M}$ be a valid stencil, and let ${\{\hat{A}^{M}(\bm{\alpha})\}_{\bm{\alpha}\in\mathcal{P}_{d}}}$ denote the corresponding ${M}$-induced valid PPO frame. The associated \emph{${M}$-Weyl--Heisenberg displacement operators} (${M}$-WHDOs) are defined by
    \begin{align}
        \hat{V}^{M}(\bm{\alpha})
        \coloneqq
        \SDFT[d][\hat{A}^{M}](\bm{\alpha}),
    \label{eq_M_WHDO_def}
    \end{align}
    where ${\bm{\alpha}\in\mathcal{P}_{d}}$.
\end{definition}
Furthermore, since ${\SDFT[d]}$ is self-inverse, one also has that
\begin{align}
    \hat{A}^{M}(\bm{\alpha})
    =
    \SDFT[d][\hat{V}^{M}](\bm{\alpha}).
\label{eq_M_PPO_from_M_WHDO}
\end{align}
Analogously to the PPO construction, the induced WHDOs are stencil-dependent reweightings of the doubled family on the base ${d\times d}$ cell.

\begin{proposition}[Explicit form of the ${M}$-WHDOs]
\label{prop_explicit_form_M_WHDOs}
    Let ${M}$ be a valid stencil and let ${\nu}$ be its associated admissible base-cell phase function. Then, for every ${\bm{\alpha}\in\mathcal{P}_{d}}$, where ${\bm{\alpha}}$ is identified with its canonical representative in the base ${d\times d}$ cell of ${\mathcal{P}_{2d}}$,
    \begin{align}
        \hat{V}^{M}(\bm{\alpha})
        &=
        2d\,\bar{\check{M}}(\bm{\alpha})^* ~\hat{V}^{(2d)}(\bm{\alpha}) \nonumber \\
        &=
        \omegaT[2d][-\nu(\bm{\alpha})] ~\hat{V}^{(2d)}(\bm{\alpha}) \nonumber \\
        &=
        \omegaT[2d][-\nu(\bm{\alpha})-\alpha_1 \alpha_2] \wzed^{\alpha_2} \wex^{\alpha_1}.
        \label{eq_M_WHDO_explicit_form_doubled}
    \end{align}
\end{proposition}
\begin{proof}
The induced PPOs depend only on ${\bar{M}=\Proj M}$ by \Cref{eq_projected_stencil_generates_same_PPO}. We may therefore replace ${M}$ by ${\bar{M}}$ before applying the defining SDFT. The comb ${\DeltaT[2]}$ selects precisely the even labels ${\bm{m}=2\bm{\beta}}$, so the ${1/(2d)}$ normalization of the doubled transform reproduces the ${1/d}$ transform of ${\frac12(\bar{M}^*\star\hat{A}^{(2d)})(2\inputdot)}$. Thus
\begin{align}
 \hat{V}^M(\bm{\alpha})
 &=\left.\SDFT[2d][\DeltaT[2][\inputdot](\bar{M}^*\star\hat{A}^{(2d)})](\bm{m})\right|_{\bm{m}=\bm{\alpha}}.
\end{align}

For functions on ${\mathcal P_T}$, let ${\pointwise}$ denote pointwise multiplication: ${(u\pointwise v)(\bm a)=u(\bm a)v(\bm a)}$, with scalar multiplication understood when one function is operator-valued. For scalar functions ${g,h:\mathcal P_T\to\complex}$ and a scalar- or operator-valued function ${f}$, finite Fourier orthogonality gives
\begin{align}
    \SDFT[T][h^*\pointwise(g^*\star f)]
    &= \check h\star\bigl(\check g(-\inputdot)\pointwise\check f\bigr).
    \label{eq_cross_corr_theore_three_funcs}
\end{align}
When ${f}$ is operator-valued, the same identity holds with the SDFT and cross-correlation interpreted as scalar-weighted sums of operators.

Next, apply \Cref{eq_cross_corr_theore_three_funcs} with ${h=\DeltaT[2]}$, ${g=\bar{M}}$ and ${f=\hat{A}^{(2d)}}$. The sums are finite and the other factors are scalars. Since ${\SDFT[2d][\DeltaT[2]]=\frac d2\DeltaT[d]}$, ${\SDFT[2d][\hat{A}^{(2d)}]=\hat{V}^{(2d)}}$ and ${\bar{\check{M}}(-\bm{m})=\bar{\check{M}}(\bm{m})^*}$, evaluating the resulting cross-correlation gives
\begin{align}
 \hat{V}^M(\bm{\alpha})
 &=\frac d2\sum_{\bm{b}\in\intsT[2][2]}
 \bar{\check{M}}(\bm{\alpha}+d\,\bm{b})^*\hat{V}^{(2d)}(\bm{\alpha}+d\,\bm{b}).
 \label{eq_appendix_M_WHDO_bit_sum}
\end{align}

To evaluate the shifted-label sum, note that the projected Fourier stencil and doubled WHDO obey the same ${d}$-shift recurrence:
\begin{align}
 \bar{\check{M}}(\bm{\alpha}+d\,\bm{b})
 &=\omegaT[2][\symprod{\bm{\alpha}}{\bm{b}}-d\,b_1b_2]\bar{\check{M}}(\bm{\alpha}),\\
 \hat{V}^{(2d)}(\bm{\alpha}+d\,\bm{b})
 &=\omegaT[2][\symprod{\bm{\alpha}}{\bm{b}}-d\,b_1b_2]\hat{V}^{(2d)}(\bm{\alpha}).
\end{align}
In each summand of \Cref{eq_appendix_M_WHDO_bit_sum}, this sign multiplies its complex conjugate and cancels. Consequently the four terms coincide, and
\begin{align}
 \hat{V}^M(\bm{\alpha})=2d\,\bar{\check{M}}(\bm{\alpha})^*\hat{V}^{(2d)}(\bm{\alpha}).
\end{align}
This proves the first form in \Cref{prop_explicit_form_M_WHDOs}; its remaining forms follow by substituting the base-cell phase normal form and \Cref{def_DV_doubled_WHDOs}.
\end{proof}

Each ${M}$-WHDO is unitary because it is a unit-modulus rephasing of a doubled WHDO. Moreover, as ${\nu(\bm0) \equiv_{2d} 0}$ by ${\criterion{\nu2}}$, we also have that ${\hat V^M(\bm0)=\idop}$.

\begin{remark}[Interpretation of ${\nu}$]
    The base-cell phase function ${\nu}$ does not merely parameterize admissible projected stencils. It is also exactly the phase function ${\nu}$ that rephases the base-cell restriction of the doubled WHDO family to produce the stencil-induced displacement family.
\end{remark}

\subsection{Associated \titm{M}-characteristic functions}
\label{subsec_M_characteristic_functions}
The $M$-WHDOs provide the operator-basis symplectic-Fourier side to the $M$-PPO frame. For a valid stencil ${M}$ and ${\hat O\in\mathcal L(\complex^d)}$, we define the stencil-induced characteristic function---the $M$-characteristic function---by taking the $d$-sized SDFT of the corresponding $M$-DWF:
\begin{align}
    \bigChi_{\hat O}^M(\bm\alpha)
    &\coloneqq \SDFT[d][W_{\hat O}^M](\bm\alpha) \nonumber\\
    &= \frac1d\traceOf{\hat V^M(\bm\alpha)}{\hat O},
    \qquad \bm\alpha\in\mathcal P_d.
    \label{eq_M_characteristic_function}
\end{align}
The trace form follows from Hermiticity \criterion{\mathrm{A}1} and \Cref{eq_M_DWF_trace_form,eq_M_WHDO_def}. Since the SDFT is invertible, $W_{\hat{O}}^M$ and $\bigChi_{\hat{O}}^M$ contain the same operator information in their respective operator representations: the $M$-DWF uses the $M$-PPO frame, whereas the characteristic function uses the $M$-WHDO family. Combining the inverse SDFT with the Hilbert--Schmidt orthogonality of the ${M}$-PPOs, \criterion{\mathrm{A}3}, gives the reconstruction formula ${\hat{O} = \sum_{\bm{\alpha}\in\mathcal{P}_{d}} \bigChi_{\hat{O}}^M(\bm{\alpha}) \hat{V}^M(\bm{\alpha})^\dagger}$.

By \Cref{prop_explicit_form_M_WHDOs}, the characteristic function can also be written as ${\bigChi_{\hat O}^M(\bm\alpha)=d^{-1}\omegaT[2d][-\nu(\bm\alpha)]\traceOf{\hat V^{(2d)}(\bm\alpha)}{\hat O}}$. Furthermore, for two valid stencils $M_1$ and $M_2$ at the same dimension, with associated phase functions $\nu_1$ and $\nu_2$, the corresponding characteristic functions are related by
\begin{align}
    \bigChi_{\hat{O}}^{M_2}(\bm{\alpha})
    =
    \omegaT[2d][-\bigl(\nu_2(\bm{\alpha})-\nu_1(\bm{\alpha})\bigr)]
    \bigChi_{\hat{O}}^{M_1}(\bm{\alpha}),
    \label{eq_M_characteristic_function_change_of_stencil_rephasing}
\end{align}
and thus, the two characteristic functions of the fixed operator ${\hat{O}}$ satisfy
\begin{align}
    \abs{\bigChi_{\hat O}^{M_2}(\bm\alpha)}
    =
    \abs{\bigChi_{\hat O}^{M_1}(\bm\alpha)},
    \qquad \bm\alpha\in\mathcal P_d.
\end{align}
Their pointwise magnitude is therefore stencil independent, while their phase, wherever the characteristic functions are nonzero, is generally stencil dependent. The admissible phase function ${\nu}$ therefore determines both the projected stencil and the rephasing of these characteristic-function values, without introducing any additional independent parameters.

\subsection{Subspace imposed by canonical horizontal and vertical marginals}
\label{subsec_subspace_imposed_by_hori_vert_marg}
We now determine the restriction imposed on valid PPO frames by horizontal and vertical marginalization. Since marginalization is an additional structural requirement beyond PPO-frame validity itself, we first define the relevant marginal operators and recall the marginalization requirement of Ref.~\cite{antonopoulos_Grand_2025}.

The discussion below is restricted to axis-parallel marginal lines. More general lines are considered in Refs.~\cite{leonhardt_Discrete_1996,miquel_Quantum_2002, horibe_Existence_2002,wootters_Wignerfunction_1987}, but the associated marginal requirements can impose additional restrictions on the qudit dimension. Such generalizations are left to future work.

\begin{definition}[PPO marginals and marginalization]
\label{def_DV_PPO_marginals}
    Given a valid PPO frame ${\{\hat{A}(\bm{\alpha})\}_{\bm{\alpha}\in\mathcal{P}_{d}}}$, define the vertical (V) and horizontal (H) \emph{PPO marginals} by
    \begin{align}
        \hat{Q}_V(\alpha_1) \coloneqq \frac{1}{d} \sum_{\alpha_2} \hat{A}(\bm{\alpha}), \\
        \qq{ and }
        \hat{Q}_H(\alpha_2) \coloneqq \frac{1}{d} \sum_{\alpha_1} \hat{A}(\bm{\alpha}),
    \end{align} 
    where these operators are averages of the PPOs along vertical and horizontal (i.e., axis-parallel) lines, respectively.
    
    \begin{itemize}[leftmargin=3.45em, labelsep=0.6em]
        \item[\criterion{\mathrm{A}5}.] (Marginalization). A valid PPO frame satisfies \emph{marginalization} if each family ${\{\hat{Q}_V(\alpha_1)\}_{\alpha_1\in\intsT[d]}}$ and ${\{\hat{Q}_H(\alpha_2)\}_{\alpha_2\in\intsT[d]}}$ consists of rank-one projectors resolving the identity, and the two associated bases are mutually unbiased~\cite{gibbons_Discrete_2004,ivonovic_Geometrical_1981,durt_Mutually_2010}. For these rank-one projective resolutions, mutual unbiasedness is equivalently expressed by
        \begin{align}
            \traceOf{\hat{Q}_V(\alpha_1)}{\hat{Q}_H(\alpha_2)} = \frac{1}{d}.
            \label{eq_cross_family_overlap}
        \end{align}
    \end{itemize}
    \par\smallskip\noindent\textbf{Further restriction for the work below.} For the work below, we further restrict to the canonical labeled marginals
    \begin{align}
        \hat Q_V^M(\alpha_1)&=\qoutprod{\alpha_1}{\alpha_1},\\
        \hat Q_H^M(\alpha_2)&=\poutprod{\alpha_2}{\alpha_2},
        \label{eq_canonical_projector_marginals}
    \end{align}
    for every ${\alpha_1,\alpha_2\in\intsT[d]}$, with the position and momentum bases fixed by \Cref{def_mom_pos_bases}.
\end{definition}

The projector requirement in \criterion{\mathrm{A}5} is essential. (The cross-family overlap relation, \Cref{eq_cross_family_overlap}, already follows from \criterion{\mathrm{A}1} and \criterion{\mathrm{A}3}.) The further restriction specifies both bases and their labels, in addition to the rank-one projectivity required by \criterion{\mathrm{A}5}. For stencil-induced frames, it is equivalent to the Fourier-axis constraints
\begin{align}
    \bar{\check{M}}(0,\alpha_2)=\frac{1}{2d},
    \qq{ and }
    \bar{\check{M}}(\alpha_1,0)=\frac{1}{2d},
    \label{eq_ch8_axis_constraints_bar_check_M}
\end{align}
for all ${\alpha_1,\alpha_2\in\intsT[d]}$. Equivalently, in terms of the base-cell phase function ${\nu}$, this becomes
\begin{align}
\label{eq_ch8_axis_constraints_nu}
    \nu(0,\alpha_2)\equiv_{2d}0,
    \qq{ and }
    \nu(\alpha_1,0)\equiv_{2d}0.
\end{align}
The following calculation establishes both directions of the equivalence.

By the inverse SDFT in \Cref{eq_M_PPO_from_M_WHDO}, summing ${\hat{A}^M(\alpha_1,\alpha_2)}$ over ${\alpha_2}$ enforces ${\kappa_1=0}$ through finite Fourier orthogonality. With the ${1/d}$ marginal normalization, this gives
\begin{align}
 \hat{Q}_V^M(\alpha_1)
 &=\frac1d\sum_{\kappa_2\in\intsT[d]}\omegaT[d][-\alpha_1\kappa_2]\hat{V}^M(0,\kappa_2)\nonumber\\
 &=2\sum_{\kappa_2\in\intsT[d]}\omegaT[d][-\alpha_1\kappa_2]\bar{\check{M}}(0,\kappa_2)^*\wzed^{\kappa_2},\\
 \hat{Q}_H^M(\alpha_2)
 &=2\sum_{\kappa_1\in\intsT[d]}\omegaT[d][\alpha_2\kappa_1]\bar{\check{M}}(\kappa_1,0)^*\wex^{\kappa_1}.
\end{align}
Here \Cref{eq_M_WHDO_explicit_form_doubled} gives the second line, and summing over ${\alpha_1}$ similarly gives the horizontal expression. The explicit-form proof uses no marginal assumption.

The position and momentum conventions of \Cref{def_mom_pos_bases} give
\begin{align}
 \qoutprod{\alpha_1}{\alpha_1}
 &=\frac1d\sum_{\kappa_2\in\intsT[d]}\omegaT[d][-\alpha_1\kappa_2]\wzed^{\kappa_2},\\
 \poutprod{\alpha_2}{\alpha_2}
 &=\frac1d\sum_{\kappa_1\in\intsT[d]}\omegaT[d][\alpha_2\kappa_1]\wex^{\kappa_1}.
\end{align}
Thus setting ${\bar{\check{M}}=1/(2d)}$ on both Fourier axes is sufficient. Conversely, equality to the canonical projectors for all labels forces equality of the finite Fourier coefficients: ${2\bar{\check{M}}^*=1/d}$ on each axis. This proves necessity as well. In the base-cell phase normal form these coefficient constraints are precisely ${\nu(0,\alpha_2)\equiv_{2d}0}$ and ${\nu(\alpha_1,0)\equiv_{2d}0}$, as claimed.

Thus, the canonical horizontal and vertical marginal conditions fix the base-cell phase function ${\nu}$ along the coordinate axes. The resulting constrained parameter space is as follows:

\begin{proposition}[Topological reduction under canonical horizontal and vertical marginals]
\label{prop_canonical_hor_vert_line_marginals}
    Let ${\solspace^{\mathrm{HV}}\subseteq\solspace}$ denote the constrained parameter space obtained by imposing the canonical labeled horizontal and vertical marginal conditions [see \Cref{def_DV_PPO_marginals}], equivalently \Cref{eq_ch8_axis_constraints_nu}, equipped with the subspace topology inherited from ${\topol}$. Then,
    \begin{align}
        \solspace^{\mathrm{HV}}
        \cong
        (S^1)^{N_c^{\mathrm{HV}}(d)}\times \intsT[2][N_b^{\mathrm{HV}}(d)],
    \end{align}
    where 
    \begin{align}
        N_c^{\mathrm{HV}}(d) &\coloneqq \floor{\frac{(d-1)^2}{2}}
        \qq{ and } \\
        N_b^{\mathrm{HV}}(d) &\coloneqq \{d-1\}_2.        
    \end{align}
    Moreover,
    \begin{align}
        N_c^{\mathrm{HV}}(d) = \floor{\frac{(d-1)^2}{2}} &\leq \frac{d^2 - 1 - 3 \{d-1\}_2}{2} \nonumber \\
        &= N_c(d),
    \end{align}
    with equality only at ${d=2}$ for qudit dimensions ${d\geq 2}$, and
    \begin{align}
        N_b^{\mathrm{HV}}(d)
        =
        \{d-1\}_2
        \leq
        3\{d-1\}_2
        =
        N_b(d),
    \end{align}
    with equality only when ${d}$ is odd, in which case
    ${N_b^{\mathrm{HV}}(d)=N_b(d)=0}$.
    
    Thus, imposing horizontal and vertical marginals never increases the number of free parameters. It strictly reduces the number of continuous parameters for all ${d>2}$, and strictly reduces the number of discrete bit parameters when ${d}$ is even.
\end{proposition}

\begin{proof}
The marginal constraints fix ${\nu=0}$ on both coordinate axes. These values satisfy twisted oddness because ${\Xi=0}$ there. Therefore, only the off-axis set
\begin{align}
 \mathcal{P}_d^{\mathrm{off}}
 \coloneqq\{(\alpha_1,\alpha_2)\in\mathcal{P}_d\mid\alpha_1\ne0,\ \alpha_2\ne0\}
\end{align}
contributes free parameters. It has ${(d-1)^2}$ points and is invariant under negation.

For odd ${d}$, the only self-negative point is the origin, which has already been fixed. All off-axis points therefore form distinct pairs. For even ${d}$, precisely one self-negative point remains off-axis, namely ${(d/2,d/2)}$; its two choices ${0,d}$ contribute one bit. All other off-axis points again form distinct pairs. Hence
\begin{align}
 N_c^{\mathrm{HV}}(d)&=\left\lfloor\frac{(d-1)^2}{2}\right\rfloor,\\
 N_b^{\mathrm{HV}}(d)&=\{d-1\}_2.
\end{align}
The restriction and reconstruction maps used in the proof of \Cref{thm_ch8_topology_of_admissible_stencils}, now applied only to these surviving coordinates, give
\begin{align}
 \solspace^{\mathrm{HV}}\cong(S^1)^{N_c^{\mathrm{HV}}(d)}\times\intsT[2][N_b^{\mathrm{HV}}(d)].
\end{align}
Thus the count identifies the inherited topology, not only the number of parameters.

Subtracting these counts from the unrestricted counts removes ${d-1}$ circles and no bits for odd ${d}$, and ${d-2}$ circles and two bits for even ${d}$. Consequently the continuous count strictly decreases for every ${d>2}$, with equality only at ${d=2}$ among qudit dimensions ${d\ge2}$. The bit count decreases exactly in even dimensions. This proves the stated inequalities and equality cases; at ${d=1}$, both spaces are the trivial singleton.
\end{proof}

This topological reduction has a simple geometric origin: the constraints \Cref{eq_ch8_axis_constraints_nu} remove all freedom along the coordinate axes, so only the off-axis points contribute to the remaining parameter count and, hence, the restricted parameter space. In particular, for odd ${d}$, no discrete bit sector survives, while for even ${d}$, exactly one does, coming from the off-axis self-negative point ${\left(\frac{d}{2},\frac{d}{2}\right)}$.

\Cref{fig_ch8_low_dim_examples_HV} illustrates this reduction relative to \Cref{fig_ch8_low_dim_examples} while \Cref{tab_ch8_HV_topology_counts} summarizes the resulting low-dimensional counts.

\begin{figure}[t]
    \centering
    \includegraphics[width=\columnwidth]{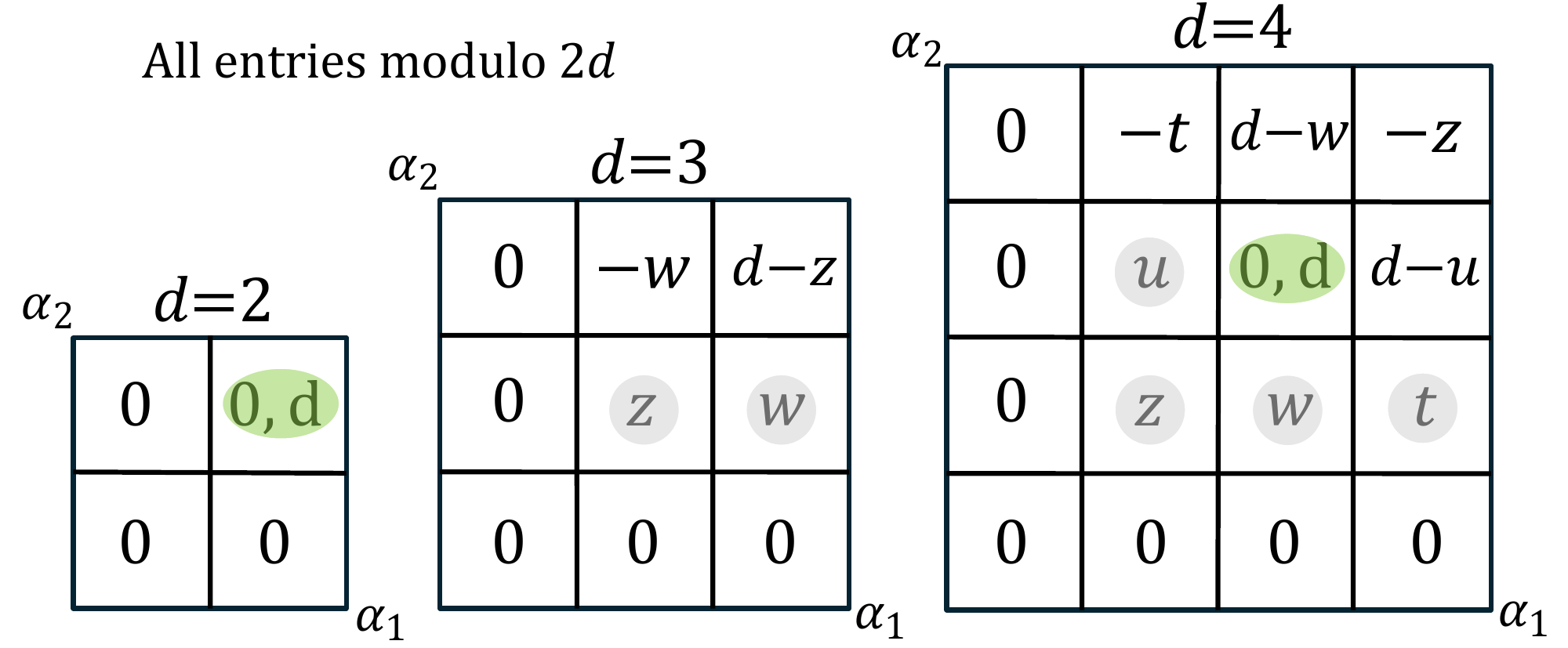}
    \caption[Free parameter structure under horizontal and vertical marginals for \titm{d=2,3,4}] {Low-dimensional examples of the free parameter structure after imposing the canonical horizontal and vertical marginal constraints. Left: ${d=2}$, where the remaining parameter space is ${\intsT[2]}$. Middle: ${d=3}$, where it is ${(S^1)^2}$. Right: ${d=4}$, where it is ${(S^1)^4\times\intsT[2]}$. The coordinate-axis values are fixed to zero by \Cref{eq_ch8_axis_constraints_nu}. Gray shadings denote the remaining independent ${S^1}$-type parameters of ${\nu}$, while green shadings denote the surviving discrete ${\intsT[2]}$-type choices at off-axis self-negative points. The displayed partners satisfy ${\nu(\{-\bm{\alpha}\}_d)\equiv_{2d}\Xi(\bm{\alpha})-\nu(\bm{\alpha})}$, with ${\Xi}$ given by \Cref{eq_ch8_Xi_def}; all entries are read modulo ${2d}$.}
    \label{fig_ch8_low_dim_examples_HV}
\end{figure}

\begin{table*}[t]
    \centering
    \renewcommand{\arraystretch}{1.2}
    \setlength{\tabcolsep}{10pt}
    \begin{tabular}{|c|c|c|c|c|c|c|c|}
        \hline
        Parameter space & ${d=1}$ & ${d=2}$ & ${d=3}$ & ${d=4}$ & ${d=5}$ & ${d=6}$ & ${d=7}$ \\
        \hline
        ${\solspace}$ & ${\{\ast\}}$ & ${\intsT[2][3]}$ & ${(S^1)^4}$ & ${(S^1)^6 \times \intsT[2][3]}$ & ${(S^1)^{12}}$ & ${(S^1)^{16} \times \intsT[2][3]}$ & ${(S^1)^{24}}$ \\
        \hline
        ${\solspace^{\mathrm{HV}}}$ & ${\{\ast\}}$ & ${\intsT[2]}$ & ${(S^1)^2}$ & ${(S^1)^4 \times \intsT[2]}$ & ${(S^1)^8}$ & ${(S^1)^{12} \times \intsT[2]}$ & ${(S^1)^{18}}$ \\
        \hline
    \end{tabular}
    \caption[Effect of horizontal and vertical marginals on the admissible parameter space]{Effect of imposing canonical horizontal and vertical marginals in the first few dimensions. The first row gives the unconstrained admissible parameter space ${\solspace}$, while ${\solspace^{\mathrm{HV}}}$ is obtained after fixing the Fourier-axis values. This removes all axis degrees of freedom---in even dimensions, only the off-axis self-negative bit survives. For ${d=1}$, both spaces are singletons.}
\label{tab_ch8_HV_topology_counts}
\end{table*}

Thus, imposing the canonical labeled horizontal and vertical projector marginals still leaves a substantial admissible family, with continuous freedom remaining for every ${d>2}$.

This canonical-marginal family can be identified with the translation-covariant part of Takami \textit{et al.}'s construction~\cite{takami_Wigner_2001}. Their operators ${\widetilde\Delta(n,m)}$ are the discrete Fourier transforms of their Fano operators [their Eq.~(27)]. Restricting to the rephased clock--shift form of their Eq.~(42), ${\widetilde\Delta_{0}(n,m)=e^{i\vartheta(n,m)}P^nS^m}$, we identify their Fano operators with the corresponding ${M}$-PPOs and set ${N=d}$, ${P=\wzed}$, ${S=\wex^{-1}}$ and ${\bm\alpha=(\{-m\}_d,n)^\tp}$. Writing their phase angle as ${\vartheta}$, the conventions are related by
\begin{align}
    \widetilde\Delta_{0}(n,m)&=\hat V^M(\bm\alpha),\nonumber\\
    e^{i\vartheta(n,m)}&=\omegaT[2d][-\nu(\bm\alpha)-\alpha_1\alpha_2].
    \label{eq_takami_phase_dictionary}
\end{align}
This identification concerns the translation-covariant sector of their construction; their full family also permits more general mixing of Fourier components.

\section{Stencil-induced Weyl--Heisenberg structure and algebraic restrictions of the admissible parameter space}
\label{sec_induced_WHDO_struct_and_algeb_restr}
The operator family introduced in \Cref{subsec_stencil_ind_M_WHDOs} makes the same base-cell phase data available for algebraic analysis. We give examples, derive the multiplication laws and determine how an order-dividing-${d}$ requirement restricts ${\solspace}$. The general-form WHDOs ${\hat V^{\mathrm{Gen}}}$ [\Cref{def_WHDO_general_phase}] are used for phase-convention-independent covariance statements.

The reduction-stencil (RS), coarse-grain-stencil (CGS) and Dirichlet-kernel-stencil (DKS) examples connect this abstract rephasing formula to the stencil constructions of Ref.~\cite{antonopoulos_Grand_2025}. RS and DKS are valid for odd ${d}$, whereas CGS is valid for even ${d}$. Additionally, direct comparison shows that, for every even ${d}$, CGS is Fourier-unitarily equivalent, up to a quarter-turn of phase-space coordinates, to Takami \textit{et al.}'s explicit even-dimensional construction~\cite{takami_Wigner_2001}.

\subsection{Example stencil-induced \titm{M}-WHDOs}
\label{subsec_example_M_WHDOS_induced_RS_CGS_DKS}\label{app_example_M_WHDOs}

The projected SDFTs of the RS, CGS and DKS from Ref.~\cite{antonopoulos_Grand_2025}, substituted into \Cref{prop_explicit_form_M_WHDOs}, give the factors ${c_M(\bm\alpha)=\omegaT[2d][-\nu(\bm\alpha)]}$ in \Cref{tab_example_M_WHDOs}. The derivations and the comparison with Gross follow below.

For odd ${d}$, the RS-induced ${M}$-WHDO coincides, under the present conventions, with Gross' displacement-operator convention~\cite{gross_Hudsons_2006}; the corresponding PPO-level equivalence was established in Sec.~VII of Ref.~\cite{antonopoulos_Grand_2025}. For the DKS, define $\indicDKS : \intsT[2d] \to \intsT[2]$,
\begin{align}
    \indicDKS(m)
&\coloneqq
    1-H(\Re\omegaT[2d][m]).
\end{align}
Here ${H}$ is the Heaviside step function, and the superscript ${\mathrm c}$ denotes the complement of the scalar boxcar function ${\Pi(m)=H(\Re\omegaT[2d][m])}$ defined in Eq.~(A19) of Ref.~\cite{antonopoulos_Grand_2025}. Also, define its vector-valued extension, $\bm{\indicDKS}:\mathcal P_{2d} \to \intsT[2][2]$,
\begin{align}
   \bm{\indicDKS}(\bm m)
&=
    \begin{pmatrix}
        \indicDKS(m_1)
    \\
        \indicDKS(m_2)
    \end{pmatrix}.
\end{align}
Its values at ${\bm\alpha\in\mathcal P_d}$ use the canonical base-cell representatives. Here ${d}$ is odd, so ${\Re\omegaT[2d][m_j]=\cos(\pi m_j/d)}$ never vanishes at an integer label; thus, no choice of ${H(0)}$ is needed. Since the DKS prefactor takes values in ${\{+1,-1\}}$, it is specifically a sign rephasing.

\begin{table}[b]
    \footnotesize
    \setlength{\tabcolsep}{0.5pt}
    \renewcommand{\arraystretch}{1.2}
    \begin{tabular}{c c c}
        \hline\hline
        Stencil
        & ${c_M(\bm{\alpha})}$
        & ${\nu(\bm{\alpha}) \equiv_{2d}}$ \\
        \hline
        RS
        &
        ${\omegaT[2d][-d\,\alpha_1\alpha_2]}$
        &
        ${d\,\alpha_1\alpha_2}$
        \\[1mm]
        CGS
        &
        ${\omegaT[2d][\symprod{\bm{\alpha}}{\{\bm{\alpha}\}_2}]}$
        &
        ${-\symprod{\bm{\alpha}}{\{\bm{\alpha}\}_2}}$
        \\[1mm]
        DKS
        &
        ${\omegaT[2][
            \symprod{\bm{\indicDKS}(\bm{\alpha})}{\bm{\alpha}}
            -\indicDKS(\alpha_1)\indicDKS(\alpha_2)]}$
        &
        ${d\,\bigl[
            \symprod{\bm{\indicDKS}(\bm{\alpha})}{\bm{\alpha}}
            -\indicDKS(\alpha_1)\indicDKS(\alpha_2)
        \bigr]}$
        \\
        \hline\hline
    \end{tabular}
    \caption{Example stencil-induced rephasing factors and corresponding base-cell phase functions.}
    \label{tab_example_M_WHDOs}
\end{table}

We derive the three entries of \Cref{tab_example_M_WHDOs} in turn. For the reduction stencil in odd dimension, applying the SDFT and projector to the stencil of Ref.~\cite{antonopoulos_Grand_2025} gives
\begin{align}
 \bar{\check{M}}_{\RS}(\bm{m})=\frac1{2d}\omegaT[2d][d\,m_1m_2].
\end{align}
Substitution into \Cref{eq_M_WHDO_explicit_form_doubled} gives ${\hat{V}^{\RS}(\bm{\alpha})=(-1)^{\alpha_1\alpha_2}\hat{V}^{(2d)}(\bm{\alpha})}$. To compare with Gross' convention, write out the doubled WHDO:
\begin{align}
 \hat{V}^{\RS}(\bm{\alpha})
 &=\omegaT[2d][(d-1)\alpha_1\alpha_2]\wzed^{\alpha_2}\wex^{\alpha_1}\nonumber\\
 &=\omegaT[d][-2^{-1}\alpha_1\alpha_2]\wzed^{\alpha_2}\wex^{\alpha_1},
 \label{eq_appH_RS_M_WHDO_final}
\end{align}
where ${2^{-1}=(d+1)/2}$ in ${\intsT[d]}$ for odd ${d}$. Thus the RS-induced family agrees with Gross' WHDO convention~\cite{gross_Hudsons_2006}.

For even ${d}$, Fourier transformation and projection of the coarse-grain stencil of Ref.~\cite{antonopoulos_Grand_2025} give
\begin{align}
 \bar{\check{M}}_{\CGS}(\bm{m})=\frac1{2d}\omegaT[2d][\symprod{\{\bm{m}\}_2}{\bm{m}}].
\end{align}
In \Cref{eq_M_WHDO_explicit_form_doubled}, complex conjugation reverses this phase; antisymmetry of ${\mat\Omega}$ therefore gives ${\hat{V}^{\CGS}(\bm{\alpha})=\omegaT[2d][\symprod{\bm{\alpha}}{\{\bm{\alpha}\}_2}]\hat{V}^{(2d)}(\bm{\alpha})}$, as stated in \Cref{tab_example_M_WHDOs}.

For odd ${d}$, Eqs.~(A17)--(A21) of Ref.~\cite{antonopoulos_Grand_2025} give the unprojected DKS SDFT as ${2/d}$ on the centered ${d\times d}$ cell and zero elsewhere. This support and its three translates by ${d\,\bm{b}}$ partition ${\mathcal{P}_{2d}}$. Thus only ${\bm{b}=\bm{\indicDKS}(\bm{m})}$ contributes in \Cref{eq_proj_explicit}, whose sign gives
\begin{align}
 \bar{\check{M}}_{\DKS}(\bm{m})
 =\frac1{2d}\omegaT[2][\symprod{\bm{\indicDKS}(\bm{m})}{\bm{m}}-\indicDKS(m_1)\indicDKS(m_2)].
\end{align}
This projected Fourier stencil is real, so the conjugation in \Cref{eq_M_WHDO_explicit_form_doubled} leaves its sign unchanged. Multiplication by ${2d\,\hat{V}^{(2d)}(\bm{\alpha})}$ therefore gives the DKS-induced family in \Cref{tab_example_M_WHDOs}.

\subsection{Basic \titm{M}-WHDO properties from \criterion{\mathrm{A}1}--\criterion{\mathrm{A}4}}
The next question is which familiar WHDO properties survive for every valid stencil, independently of any further algebraic restrictions on ${\nu}$. The following proposition records the properties inherited directly from the \criterion{\mathrm{A}1}--\criterion{\mathrm{A}4} validity criteria of the induced ${M}$-PPO frame.

\begin{proposition}[Basic properties of the ${M}$-WHDOs]
\label{prop_basic_M_WHDO_properties}
    Let ${M}$ be a valid stencil. Addition and negation of phase-space labels below are understood modulo ${d}$. Then, for all ${\bm{\alpha},\bm{\beta}\in\mathcal{P}_{d}}$, the associated ${M}$-WHDO family satisfies:
    \begin{itemize}
    \item[\criterion{\mathrm{V}1}.]\emph{(Adjoint relation).} The ${M}$-WHDOs satisfy
    \begin{align}
    \label{eq_M_WHDO_V1}
        \hat{V}^{M}(\bm{\alpha})^\dagger
        =
        \hat{V}^{M}(-\bm{\alpha}).
    \end{align}
    Furthermore, in the ${d=2}$ case, the ${M}$-WHDOs are Hermitian.

    \item[\criterion{\mathrm{V}2}.]\emph{(Traciality).} The ${M}$-WHDOs satisfy
    \begin{align}
    \label{eq_M_WHDO_V2}
        \traceOf{\hat{V}^{M}(\bm{\alpha})}{}
        =
        d\,\DeltaT[d][\bm{\alpha}].
    \end{align}

    \item[\criterion{\mathrm{V}3}.]\emph{(Hilbert--Schmidt orthogonality).} The ${M}$-WHDOs form a Hilbert--Schmidt orthogonal operator basis:
    \begin{align}
    \label{eq_M_WHDO_V3}
        \traceOf{\hat{V}^{M}(\bm{\alpha})^\dagger}{\hat{V}^{M}(\bm{\beta})}
        =
        d\,\DeltaT[d][\bm{\alpha}-\bm{\beta}].
    \end{align}

    \item[\criterion{\mathrm{V}4}.]\emph{(WHDO-covariance).} The ${M}$-WHDOs act as translations on the ${d\times d}$ phase space, translating an ${M}$-PPO phase-space label:
    \begin{align}
    \label{eq_M_WHDO_V4}
        \hat{V}^{M}(\bm{\beta})
        \hat{A}^{M}(\bm{\alpha})
        \hat{V}^{M}(\bm{\beta})^\dagger
        =
        \hat{A}^{M}(\bm{\alpha}+\bm{\beta}).
    \end{align}
    \end{itemize}
\end{proposition}
\begin{proof}
We prove \criterion{\mathrm{V}1}--\criterion{\mathrm{V}4} in turn.
For \criterion{\mathrm{V}1}, taking the adjoint of the defining SDFT and using the Hermiticity of each PPO \criterion{\mathrm{A}1} conjugates the Fourier kernel. This reverses its label:
\begin{align}
 \hat{V}^M(\bm{\alpha})^\dagger
 &=\frac1d\sum_{\bm\beta\in\mathcal P_d}\omegaT[d][\symprod{\bm\alpha}{\bm\beta}]\hat A^M(\bm\beta)\nonumber\\
 &=\hat V^M(-\bm\alpha).
\end{align}
Here negation is modulo ${d}$. At ${d=2}$, every label is self-negative, so all ${M}$-WHDOs are Hermitian.

For \criterion{\mathrm{V}2}, \criterion{\mathrm{A}2} makes every PPO trace equal to one. Linearity and finite Fourier orthogonality therefore give
\begin{align}
 \traceOf{\hat{V}^M(\bm{\alpha})}{}
 &=\frac1d\sum_{\bm{\beta}\in\mathcal{P}_d}\omegaT[d][-\symprod{\bm{\alpha}}{\bm{\beta}}]\nonumber\\
 &=d\,\DeltaT[d][\bm{\alpha}].
\end{align}

For \criterion{\mathrm{V}3}, the trace of a clock--shift monomial vanishes unless both exponents are zero modulo ${d}$. Consequently the doubled WHDOs restricted to canonical base-cell labels are Hilbert--Schmidt orthogonal. Substituting \Cref{eq_M_WHDO_explicit_form_doubled} gives
\begin{align}
 &\traceOf{\hat{V}^M(\bm{\alpha})^\dagger}{\hat{V}^M(\bm{\beta})}\nonumber\\
 &\quad=4d^2\,\bar{\check{M}}(\bm{\alpha})\bar{\check{M}}(\bm{\beta})^*
 d\,\DeltaT[d][\bm{\alpha}-\bm{\beta}]\nonumber\\
 &\quad=d\,\DeltaT[d][\bm{\alpha}-\bm{\beta}].
\end{align}
The delta forces equality of the canonical labels, where ${4d^2\,|\bar{\check{M}}(\bm{\alpha})|^2=1}$ by ${\criterion{\check{\mathrm{M}}3}}$.

For \criterion{\mathrm{V}4}, ${\hat{V}^M(\bm{\beta})}$ differs from the general-form WHDO at the same label only by a unit-modulus scalar. That scalar cancels under conjugation. Therefore \criterion{\mathrm{A}4} gives
\begin{align}
 \hat{V}^M(\bm{\beta})\hat{A}^M(\bm{\alpha})\hat{V}^M(\bm{\beta})^\dagger
 =\hat{A}^M(\bm{\alpha}+\bm{\beta}),
\end{align}
independently of the displacement phase convention. This proves \criterion{\mathrm{V}4} with addition understood in ${\mathcal{P}_d}$.
\end{proof}

\subsection{\titm{M}-WHDO multiplication laws}
\label{subsec_mult_law_and_stencil_coboundary}
The commutation rule is inherited unchanged, whereas the ordered-product law contains both the doubled-WHDO multiplier, written in base-cell labels, and a stencil-dependent phase determined by ${\nu}$.

The clock--shift relation ${\wex^{\alpha_1}\wzed^{\alpha_2}=\omegaT[d][-\alpha_1\alpha_2]\wzed^{\alpha_2}\wex^{\alpha_1}}$, together with \Cref{def_DV_doubled_WHDOs}, gives
\begin{align}
 \hat{V}^{(2d)}(\bm{\alpha})\hat{V}^{(2d)}(\bm{\beta})
 =\omegaT[2d][-\symprod{\bm{\alpha}}{\bm{\beta}}]\hat{V}^{(2d)}(\bm{\alpha}+\bm{\beta}).
 \label{eq_reference_doubled_product}
\end{align}
The output ${\bm\alpha+\bm\beta}$ in \Cref{eq_reference_doubled_product} is initially a doubled-lattice label. Reducing it to the base cell introduces a carry phase; the stencil rephasing then changes the multiplier as recorded below.

\begin{proposition}[${M}$-WHDO commutation and ordered-product laws]
\label{prop_M_WHDO_commutation_combination}
    Let ${M}$ be a valid stencil with associated base-cell phase function ${\nu}$. For ${\bm{\alpha},\bm{\beta}\in\mathcal{P}_{d}}$, take their canonical integer representatives and write their ordinary sum as
    \begin{align}
        \bm{\alpha}+\bm{\beta}
        &=
        \bm{\gamma}+d\,\bm{b},
        \nonumber \\
        \bm{\gamma}
        &\coloneqq
        \{\bm{\alpha}+\bm{\beta}\}_{d}
        \in\intsT[d][2],
        \\
        \bm{b}
        &\coloneqq
        \frac{\bm{\alpha}+\bm{\beta}-\bm{\gamma}}{d}
        \in\intsT[2][2].
    \end{align}
    Then
    \begin{align}
        \hat{V}^{M}(\bm{\alpha})\hat{V}^{M}(\bm{\beta})
        &=
        \omegaT[d][-\symprod{\bm{\alpha}}{\bm{\beta}}]
        ~\hat{V}^{M}(\bm{\beta})\hat{V}^{M}(\bm{\alpha}),
        \label{eq_M_WHDO_commutation}
        \\
        \hat{V}^{M}(\bm{\alpha})\hat{V}^{M}(\bm{\beta})
        &=
        \omegaT[2d][D_{\nu}(\bm{\alpha},\bm{\beta})+\theta_c(\bm{\gamma},\bm{b})
        -\symprod{\bm{\alpha}}{\bm{\beta}}]\hat{V}^{M}(\bm{\gamma}),
        \label{eq_M_WHDO_combination}
    \end{align}
    where
    \begin{align}
        D_{\nu}(\bm{\alpha},\bm{\beta})
        &\coloneqq
        \nu(\bm{\gamma})-\nu(\bm{\alpha})-\nu(\bm{\beta})
        \label{eq_ch8_Dnu_def}
    \end{align}
    is the change in the ordered-product phase caused by the stencil rephasing, and ${\theta_c\in\{0,d\}}$ is the canonical carry phase defined in \Cref{eq_ch8_theta_c_def}.
\end{proposition}
\begin{proof}
For the commutation law, \Cref{eq_M_WHDO_explicit_form_doubled} makes each ${M}$-WHDO a doubled WHDO multiplied by the scalar ${\omegaT[2d][-\nu(\bm{\alpha})]}$. These scalars cancel when a product is compared with its reverse. Comparing the two orders in \Cref{eq_reference_doubled_product} therefore proves \Cref{eq_M_WHDO_commutation}.

For the ordered product, the output label must be reduced to the base cell without discarding its phase. Take the canonical integer representatives and write ${\bm{\alpha}+\bm{\beta}=\bm{\gamma}+d\,\bm{b}}$ as in \Cref{prop_M_WHDO_commutation_combination}. The doubled-WHDO definition and ${\wzed^d=\wex^d=\idop}$ give
\begin{align}
 \hat{V}^{(2d)}(\bm{\gamma}+d\,\bm{b})
 &=\omegaT[2][-\gamma_1b_2-\gamma_2b_1-d\,b_1b_2]\hat{V}^{(2d)}(\bm{\gamma})\nonumber\\
 &=\omegaT[2d][\theta_c(\bm{\gamma},\bm{b})]\hat{V}^{(2d)}(\bm{\gamma}).
\end{align}
The second equality uses ${-\gamma_1b_2\equiv_2\gamma_1b_2}$ and the definition of ${\theta_c}$; this is the canonical carry phase.

The two input rephasings contribute ${-\nu(\bm{\alpha})-\nu(\bm{\beta})}$, while converting the output from ${\hat{V}^{(2d)}(\bm{\gamma})}$ to ${\hat{V}^M(\bm{\gamma})}$ contributes ${\nu(\bm{\gamma})}$. Their sum is ${D_\nu(\bm{\alpha},\bm{\beta})}$. Combining these contributions gives
\begin{align}
 \hat{V}^M(\bm{\alpha})\hat{V}^M(\bm{\beta})
 &=\omegaT[2d][D_\nu(\bm{\alpha},\bm{\beta})+\theta_c(\bm{\gamma},\bm{b})
-\symprod{\bm{\alpha}}{\bm{\beta}}]\hat{V}^M(\bm{\gamma}),
\end{align}
which proves \Cref{eq_M_WHDO_combination}. The symplectic product is evaluated on the integer representatives before reducing its contribution modulo ${2d}$.
\end{proof}

The stencil-induced rephasing [$\hat{V}^{(2d)}(\bm{\alpha}) \mapsto \hat{V}^{M}(\bm{\alpha}) = \omegaT[2d][-\nu(\bm{\alpha})]\hat{V}^{(2d)}(\bm{\alpha})$; \Cref{eq_M_WHDO_explicit_form_doubled}] explains why the commutation law is unchanged while the ordered-product phase acquires the \emph{coboundary}%
\footnote{\samepage\label{fn_coboundary_rephasing}Here, \emph{coboundary} refers to the change in an ordered-product multiplier induced by a phase redefinition. If ${\hat V(\bm\alpha)\mapsto g(\bm\alpha)\hat V(\bm\alpha)}$, then the multiplier is multiplied by ${g(\bm\alpha)g(\bm\beta)/g(\bm\gamma)}$, where ${\bm\gamma=\{\bm\alpha+\bm\beta\}_d}$. In the present case, ${g(\bm\alpha)=\omegaT[2d][-\nu(\bm\alpha)]}$, so the coboundary factor is ${\omegaT[2d][D_\nu(\bm\alpha,\bm\beta)]}$. See also the change of multiplier under rephasing in Eq.~(13) of Ref.~\cite{raussendorf_Role_2023}.} %
term~${D_\nu}$.

One might seek ${D_\nu=0}$ to recover the ordered-product multiplier of the chosen doubled-WHDO reference, including its base-cell carry phase. This would require ${\nu}$ to be additive on ${\mathcal P_d}$, meaning ${\nu(\{\bm\alpha+\bm\beta\}_d)\equiv_{2d}\nu(\bm\alpha)+\nu(\bm\beta)}$. Hence ${\nu}$ would be ordinarily odd modulo ${2d}$. Whenever the prescribed twist ${\Xi}$ is nonzero, this conflicts with twisted oddness, so ${D_\nu=0}$ is not an available simplification within the admissible family. This fixed-reference compatibility condition is distinct from the cohomological obstructions to Clifford covariance and positive Pauli-measurement representations studied in Ref.~\cite{raussendorf_Role_2023}.

\begin{remark}[Carry phases from reducing doubled WHDO labels]
    The canonical carry phase term ${\theta_c\in\{0,d\}}$ in \Cref{prop_M_WHDO_commutation_combination} comes from reducing a doubled WHDO label back to the ${d\times d}$ base cell. With ${\bm{\alpha}+\bm{\beta}=\bm{\gamma}+d\,\bm{b}}$, as in \Cref{prop_M_WHDO_commutation_combination}, the ${d}$-quasiperiodicity of the doubled WHDOs gives
    \begin{align}
        \hat{V}^{(2d)}(\bm{\gamma}+d\,\bm{b})
        =
        \omegaT[2d][\theta_c(\bm{\gamma},\bm{b})] 
        ~\hat{V}^{(2d)}(\bm{\gamma}).
    \end{align}
    This is the ${\theta_c}$ term because
    \begin{align}
        \omegaT[2][\symprod{-\bm{b}}{\bm{\gamma}}-d\,b_1b_2]
        =
        \omegaT[2d][d\,\{\symprod{\bm{\gamma}}{\bm{b}}-b_1b_2\{d\}_2\}_2]
        =
        \omegaT[2d][\theta_c(\bm{\gamma},\bm{b})].
    \end{align}
    If no component carries, then ${\bm{b}=\bm{0}}$, ${\theta_c(\bm{\gamma},\bm{b})=0}$, and \Cref{eq_M_WHDO_combination} reduces to the simpler no-carry form. Hence, ${\theta_c}$ is the fixed carry phase inherited from the doubled WHDO convention. By contrast, ${D_\nu}$ is the stencil-dependent rephasing coboundary term.

    This is closely related to the convention of Galetti and Toledo Piza~\cite{galetti_Discrete_1992}, which effectively works with labels reduced modulo ${d}$ from the outset. After translating their clock--shift convention to ours, the corresponding periodic family is obtained by evaluating ${\hat V^{(2d)}}$ at ${\{\bm\kappa\}_d}$. In the present doubled-WHDO convention, reducing an ordinary sum of labels to the base cell produces the carry phase displayed above. Thus, strict periodicity of the individual operators does not remove the carry phase from their ordered products; this phase remains distinct from the stencil-dependent ${D_\nu}$.
\end{remark}

\subsection{Optional order-\titm{d} restriction on \titm{M}-WHDOs}
\label{subsec_M_WHDO_order_d_restriction}
We now fix the Hilbert-space dimension ${d}$ and ask which admissible stencils make every induced displacement satisfy ${\bigl(\hat V^M(\bm\alpha)\bigr)^d=\idop}$. Equivalently, each operator has finite order dividing this fixed ${d}$; its exact order may be a proper divisor. Although an operator of order dividing ${d}$ also has order dividing any multiple of ${d}$, those multiples do not change the Hilbert-space dimension or the label space being classified here. We refer to this as the \emph{order-${d}$ restriction}.

The restriction asks which admissible stencils make ${d}$ repeated applications of an ${M}$-WHDO return to ${\idop}$ without a residual scalar phase, so that the return of phase-space labels after ${d}$ translations is also realized exactly at the operator level. Order-dividing-${d}$ conventions also underlie Gross' odd-dimensional Weyl operators and Raussendorf \textit{et al.}'s phase-constrained Pauli representatives~\cite{gross_Hudsons_2006,raussendorf_Role_2023}.

The doubled WHDO family is naturally indexed on the doubled lattice and has order dividing ${2d}$, but it need not have order dividing ${d}$. Specifically,
\begin{align}
    \left(\hat{V}^{(2d)}(\bm{m})\right)^{2d}
    &=\idop,
    \qquad \bm m\in\mathcal P_{2d}.
    \label{eq_dbld_WHDO_order}
\end{align}
To determine its ${d}$th power, we use the clock--shift relation
\begin{align*}
    \wex^{m_1}\wzed^{m_2}
    &=\omegaT[d][-m_1m_2]\wzed^{m_2}\wex^{m_1}.
\end{align*}
Reordering the ${d}$ clock--shift pairs requires ${d\,(d-1)/2}$ such commutations. Thus, since ${\wzed^d=\wex^d=\idop}$,
    \begin{align}
        \left(\hat{V}^{(2d)}(\bm{m})\right)^d
        &=
        \omegaT[2d][-d\,m_1m_2]
        \omegaT[d][-\frac{d\,(d-1)}{2}m_1m_2]
        \wzed^{d\,m_2}\wex^{d\,m_1} \nonumber \\
        &=
        \omegaT[2][-d\,m_1m_2]\idop.
        \label{eq_dbld_WHDO_order2}
    \end{align}
The residual sign can therefore prevent the ${d}$th power from being the identity, while squaring it gives the ${2d}$th-power identity above.

\begin{proposition}[Power and order condition for the ${M}$-WHDOs]
\label{cor_M_WHDO_order_d_condition}
    Let ${M}$ be a valid stencil with admissible base-cell phase function ${\nu}$. For every ${\bm\alpha\in\mathcal P_d}$,
    \begin{align}
        \bigl(\hat{V}^{M}(\bm{\alpha})\bigr)^d
        &=
        \omegaT[2][-\nu(\bm{\alpha})-d\,\alpha_1\alpha_2]
        ~\idop.
        \label{eq_M_WHDO_order_d_phase}
    \end{align}
    Hence the family has order dividing ${d}$, meaning ${\bigl(\hat V^M(\bm\alpha)\bigr)^d=\idop}$ for every label, if and only if
    \begin{align}
    \label{eq_M_WHDO_order_d_condition}
        \nu(\bm{\alpha})
        \equiv_2
        d\,\alpha_1\alpha_2
        \equiv_2
        \begin{cases}
        0, & d \text{ even}, \\
        \alpha_1\alpha_2, & d \text{ odd}.
    \end{cases}        
    \end{align}    
    Equivalently, the allowed values lie in the additive coset
    \begin{align}
        \nu(\bm\alpha)\in d\,\alpha_1\alpha_2+2\intsT[d],
        \label{eq_order_d_phase_coset}
    \end{align}
    where $2\intsT[d]
    =\{0,2,\ldots,2d-2\}$ and addition is modulo~$2d$.
\end{proposition}

\begin{proof}
For a canonical base-cell representative ${\bm\alpha}$, \Cref{eq_M_WHDO_explicit_form_doubled} writes $\hat V^M(\bm\alpha)=\omegaT[2d][-\nu(\bm\alpha)]\hat V^{(2d)}(\bm\alpha)$. Raising this expression to the ${d}$th power multiplies the inherited phase in \Cref{eq_dbld_WHDO_order2} by ${\omegaT[2][-\nu(\bm\alpha)]}$, giving \Cref{eq_M_WHDO_order_d_phase}. The result equals ${\idop}$ exactly when ${-\nu(\bm\alpha)-d\,\alpha_1\alpha_2\equiv_2 0}$, equivalently ${\nu(\bm\alpha)\equiv_2d\,\alpha_1\alpha_2}$. Within the ${2d}$-periodic phase circle, this congruence selects precisely the ${d}$ residues in \Cref{eq_order_d_phase_coset}.
\end{proof}

The inherited factor in \Cref{eq_M_WHDO_order_d_phase} is trivial for even ${d}$ and equals ${\omegaT[2][-\alpha_1\alpha_2]}$ for odd ${d}$. The order restriction selects even residues for even ${d}$ and residues of parity ${\alpha_1\alpha_2}$ for odd ${d}$. Opposite labels remain related by twisted oddness, and self-negative points must still satisfy their admissibility constraints; the choices are not independent at every label.

This optional parity-level restriction selects a subspace of ${\solspace}$ without changing the \criterion{\mathrm{A}1}--\criterion{\mathrm{A}4} admissibility criteria. Its topological effect is recorded next.

\begin{proposition}[Topological reduction after imposing the order-${d}$ restriction]
\label{prop_order_d_restricted_topology}    
    Let ${\solspace^{\mathrm{ord}}\subseteq\solspace}$ denote the subspace of admissible base-cell phase functions satisfying the order-${d}$ condition of \Cref{cor_M_WHDO_order_d_condition}, equipped with the subspace topology inherited from ${\topol}$. Then
    \begin{align}
        \solspace^{\mathrm{ord}}
        \cong
        \intsT[d][N_c(d)]
        \times
        \intsT[2][N_b(d)],
    \end{align}
    where ${N_c(d)}$ and ${N_b(d)}$ are the counting functions of \Cref{thm_ch8_topology_of_admissible_stencils}. Equivalently,
    \begin{align}
        \solspace^{\mathrm{ord}}
        \cong
        \begin{cases}
            \intsT[d][\frac{d^2-1}{2}],
            & d \text{ odd},
            \\[1mm]
            \intsT[d][\frac{d^2-4}{2}]
            \times
            \intsT[2][3],
            & d \text{ even}.
        \end{cases}
    \end{align}
\end{proposition}

\begin{proof}
We use the same negation-pair decomposition as in the proof of \Cref{thm_ch8_topology_of_admissible_stencils}. First we determine the choices allowed on a distinct pair; then we check that the self-negative choices survive.

The order-${d}$ condition selects the coset in \Cref{eq_order_d_phase_coset}. At a chosen representative this can be written as
\begin{align}
 \nu(\bm{\alpha})\in\{d\,\alpha_1\alpha_2+2r\mid r\in\intsT[d]\}
 \subset\reals/(2d\,\intsT).
 \label{eq_appH_order_d_allowed_values}
\end{align}
These are exactly ${d}$ distinct values: two indices give the same phase precisely when they agree modulo ${d}$. Twisted oddness then determines ${\nu(\bm{\alpha}^-)}$ by ${\nu(\bm{\alpha}^-)\equiv_{2d}\Xi(\bm{\alpha})-\nu(\bm{\alpha})}$, where ${\bm{\alpha}^-=\{-\bm{\alpha}\}_d}$.

This choice is allowed only if its forced partner also has the required parity. Reducing twisted oddness modulo ${2}$ shows that this is equivalent to
\begin{align}
 d\,\alpha_1^-\alpha_2^-
 \equiv_2\Xi(\bm{\alpha})-d\,\alpha_1\alpha_2.
 \label{eq_appH_Xi_parity_identity}
\end{align}
For even ${d}$, both sides vanish modulo ${2}$, since ${\Xi\in\{0,d\}}$. For odd ${d}$, they also vanish on either coordinate axis. Off-axis, ${\alpha_i^-=d-\alpha_i\equiv_2 1-\alpha_i}$ and \Cref{eq_ch8_Xi_def} gives ${\Xi(\bm{\alpha})\equiv_2\alpha_1+\alpha_2+1}$. Therefore,
\begin{align}
 d\,\alpha_1^-\alpha_2^-
 &\equiv_2(1-\alpha_1)(1-\alpha_2)\nonumber\\
 &\equiv_2\alpha_1+\alpha_2+1-\alpha_1\alpha_2\nonumber\\
 &\equiv_2\Xi(\bm{\alpha})-d\,\alpha_1\alpha_2.
\end{align}
Thus the partner automatically satisfies the condition, and each distinct negation pair contributes one ${\intsT[d]}$-indexed choice.

At a self-negative point, twisted oddness is a self-consistency condition, ${2\nu(\bm{\alpha})\equiv_{2d}\Xi(\bm{\alpha})}$. The origin remains fixed. Every other self-negative point occurs in even dimension and has ${\Xi=0}$, so its two admissible values are ${0}$ and ${d}$. Both are even and therefore satisfy ${\nu(\bm{\alpha})\equiv_2d\,\alpha_1\alpha_2}$. No bit is removed.

Consequently each original circle factor becomes a ${d}$-point subset, while the non-origin self-negative factors remain unchanged. These finite subsets of the phase circle carry the discrete topology, so the restriction and reconstruction maps used in the proof of \Cref{thm_ch8_topology_of_admissible_stencils} give
\begin{align}
 \solspace^{\mathrm{ord}}\cong\intsT[d][N_c(d)]\times\intsT[2][N_b(d)].
\end{align}
This proves \Cref{prop_order_d_restricted_topology}, including the singleton at ${d=1}$ and the three surviving bits at ${d=2}$.
\end{proof}

Imposing the order-${d}$ restriction does not change the negation-pair counting from \Cref{thm_ch8_topology_of_admissible_stencils}. Rather, it selects the admissible conventions whose induced displacement operators all have order dividing ${d}$. As shown in the proof, the ${\intsT[d]}$ factors arise because each former circle-valued parameter associated with a distinct negation pair is reduced to a dit-valued choice. For an independent non-self-negative negation-pair representative ${\bm{\alpha}}$, this dit may be written as ${r_{\bm{\alpha}}\in\intsT[d]}$ and embedded into the original ${2d}$-periodic phase circle by
\begin{align}
    r_{\bm{\alpha}}
    \longmapsto
    d\,\alpha_1\alpha_2+2r_{\bm{\alpha}}
    \qquad
    \text{modulo }2d.
\end{align}
Equivalently, the allowed phase representatives are ${\{0,2,\ldots,2d-2\}}$ when ${d\,\alpha_1\alpha_2}$ is even and ${\{1,3,\ldots,2d-1\}}$ when it is odd.
Therefore, each former ${S^1}$ parameter is replaced by ${d}$ equally spaced points on the same ${2d}$-periodic phase circle, \({S^1 \longmapsto \intsT[d]}\),
while the non-origin self-negative ${\intsT[2]}$ factors in even dimensions are left unchanged.

The order-${d}$ restriction may also be imposed jointly with the horizontal--vertical marginal restrictions, since both impose compatible pointwise conditions on the same base-cell phase function ${\nu}$.

\begin{table*}[t]
    \centering
    \renewcommand{\arraystretch}{1.2}
    \setlength{\tabcolsep}{10pt}
    \begin{tabular}{|c|c|c|c|c|c|c|c|}
    \hline
    Parameter space & ${d=1}$ & ${d=2}$ & ${d=3}$ & ${d=4}$ & ${d=5}$ & ${d=6}$ & ${d=7}$ \\
    \hline
    ${\solspace}$
    & ${\{\ast\}}$
    & ${\intsT[2][3]}$
    & ${(S^1)^4}$
    & ${(S^1)^6 \times \intsT[2][3]}$
    & ${(S^1)^{12}}$
    & ${(S^1)^{16} \times \intsT[2][3]}$
    & ${(S^1)^{24}}$ \\
    \hline
    ${\solspace^{\mathrm{HV}}}$
    & ${\{\ast\}}$
    & ${\intsT[2]}$
    & ${(S^1)^2}$
    & ${(S^1)^4 \times \intsT[2]}$
    & ${(S^1)^8}$
    & ${(S^1)^{12} \times \intsT[2]}$
    & ${(S^1)^{18}}$ \\
    \hline
    ${\solspace^{\mathrm{ord}}}$
    & ${\{\ast\}}$
    & ${\intsT[2][3]}$
    & ${\intsT[3][4]}$
    & ${\intsT[4][6] \times \intsT[2][3]}$
    & ${\intsT[5][12]}$
    & ${\intsT[6][16] \times \intsT[2][3]}$
    & ${\intsT[7][24]}$ \\
    \hline
    ${\solspace^{\mathrm{HV,ord}}}$
    & ${\{\ast\}}$
    & ${\intsT[2]}$
    & ${\intsT[3][2]}$
    & ${\intsT[4][4] \times \intsT[2]}$
    & ${\intsT[5][8]}$
    & ${\intsT[6][12] \times \intsT[2]}$
    & ${\intsT[7][18]}$ \\
    \hline
\end{tabular}
    \caption[Comparison of restricted admissible parameter spaces]{Comparison of the admissible parameter space ${\solspace}$ with several restricted parameter spaces for dimensions ${d=1,\ldots,7}$. The HV constraints remove the coordinate-axis degrees of freedom. The optional order-${d}$ restriction leaves the number of pairs related by negation unchanged but replaces each continuous ${S^1}$ factor by a finite ${\intsT[d]}$ factor. The jointly restricted space ${\solspace^{\mathrm{HV,ord}}}$ shows that the HV marginal constraints and the optional order-${d}$ restriction are compatible---the former remove coordinate-axis freedoms, while the latter discretizes the remaining off-axis circle factors.}
\label{tab_ch8_topology_HV_D_nu_ord_d_counts}
\end{table*}

\begin{proposition}[Joint HV-marginal and order-${d}$ restricted topological type]
\label{prop_HV_order_d_restricted_topology}
    Let
    ${\solspace^{\mathrm{HV,ord}}\coloneqq\solspace^{\mathrm{HV}}\cap\solspace^{\mathrm{ord}}}$.
    Then
    \begin{align}
        \solspace^{\mathrm{HV,ord}}
        \cong
        \intsT[d][N_c^{\mathrm{HV}}]
        \times
        \intsT[2][N_b^{\mathrm{HV}}],
    \end{align}
    where
    \begin{align}
        N_c^{\mathrm{HV}}
        =
        \left\lfloor \frac{(d-1)^2}{2} \right\rfloor,
        \qquad
        N_b^{\mathrm{HV}}
        =
        \{d-1\}_2 .
    \end{align}
\end{proposition}

\begin{proof}
The HV constraints fix ${\nu=0}$ on the coordinate axes. Before applying the order-${d}$ restriction to the surviving degrees of freedom, we must check compatibility with these already fixed values. On either axis ${\alpha_1\alpha_2=0}$, so ${\nu=0}$ automatically satisfies ${\nu(\bm{\alpha})\equiv_2d\,\alpha_1\alpha_2}$.

Therefore, the proof of \Cref{prop_order_d_restricted_topology} applies to the surviving off-axis coordinates: each of the ${N_c^{\mathrm{HV}}(d)}$ circle factors is replaced by its ${d}$ allowed values, while the possible even-dimensional self-negative bit at ${(d/2,d/2)}$ remains unchanged. Thus
\begin{align}
 \solspace^{\mathrm{HV,ord}}
 &=\solspace^{\mathrm{HV}}\cap\solspace^{\mathrm{ord}}\nonumber\\
 &\cong\intsT[d][N_c^{\mathrm{HV}}(d)]\times\intsT[2][N_b^{\mathrm{HV}}(d)],
\end{align}
with the discrete topology inherited from the phase coordinates. Both orders of imposition enforce the same two sets of equations on ${\solspace}$, so they define the same intersection. Hence the order of imposition does not affect the resulting restricted parameter space.
\end{proof}

Hence, the two restrictions are compatible---imposing HV first removes the axis degrees of freedom before discretization, whereas imposing order-${d}$ first discretizes those same axis degrees of freedom before HV selects their zero values. In either order, the resulting restricted space is the same intersection ${\solspace^{\mathrm{HV}}\cap\solspace^{\mathrm{ord}}}$.

\Cref{tab_ch8_topology_HV_D_nu_ord_d_counts} collects the distinct effects of the canonical HV and order-${d}$ conditions, together with the joint HV and order-${d}$ restriction. \Cref{fig_tori} illustrates these successive restrictions for ${d=3}$.

\begin{figure}[h]
    \centering
    \includegraphics[width=\columnwidth]{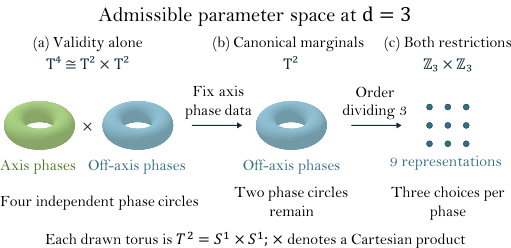}
    \caption[Tori construction for \titm{d=3}]{Effect of additional restrictions on the admissible parameter space at $d = 3$. Each drawn torus represents two independent circle-valued phase parameters. The unrestricted family is the Cartesian product of the axis and off-axis factors. Canonical horizontal and vertical marginals fix the axis phase data, leaving the off-axis torus. Requiring the associated displacements also to have order dividing 3 restricts each remaining phase to three choices, leaving nine isolated representations. Each final point denotes an entire labeled DWF convention, including the RS convention corresponding to Gross’ construction. The DKS convention belongs to the canonical-marginal torus but is excluded by the order restriction.}
    \label{fig_tori}
\end{figure}

\section{Discussion and conclusions}
\label{sec_disc_and_conc}
The stencil theorem of Ref.~\cite{antonopoulos_Grand_2025} encompasses every labeled single-qudit ${d\times d}$ discrete Wigner representation satisfying \criterion{\mathrm{A}1}--\criterion{\mathrm{A}4}. Here, we classified the valid projected stencils, and hence these representations, through their admissible base-cell phase data ${\nu}$. The parameterization separates circle-valued choices on distinct negation pairs from binary choices at non-origin self-negative points. The odd-dimensional family is a single torus, while the even-dimensional family is a disjoint union of eight tori. Thus, any two conventions in odd dimensions can be connected by continuously varying the admissible phase data, whereas continuous variation within the even-dimensional valid family cannot change the three binary choices.

Additional constraints restrict the available valid DWF representations and change the topology of their parameter space, as summarized in \Cref{tab_ch8_topology_HV_D_nu_ord_d_counts}. The canonical marginal restrictions fix the base-cell phase data on the coordinate axes, but continuous freedom remains for every ${d>2}$, and one off-axis bit survives in even dimensions. Therefore, fixing the position and momentum marginals does not determine how the full operator information is represented on phase space. The order restriction has a different effect: it replaces each independent phase circle by ${d}$ equally spaced choices while retaining the self-negative bits, making the family finite and discrete. The two restrictions are compatible, with the order restriction discretizing the freedom left by the marginals. These familiar marginal and algebraic requirements therefore constrain the choice of representation without generally selecting a unique one.

Horibe \textit{et al.}~\cite{horibe_Existence_2002} illustrate how stronger requirements can select a unique representation or exclude all solutions. Under their coordinate-marginal and determinant-one integer lattice-covariance assumptions, the solution is unique for odd ${d}$ and absent for even ${d}$. Their odd-dimensional solution coincides with Cohendet \textit{et al.}'s construction~\cite{cohendet_Stochastic_1988}, previously identified with the RS PPO up to a quadrature sign flip~\cite{antonopoulos_Grand_2025}. We do not prove an analogous collapse theorem within the stencil formalism here; the comparison shows why uniqueness or nonexistence must be assessed relative to the additional requirements imposed.

The family satisfying the canonical horizontal and vertical marginal conditions coincides with the translation-covariant part of Takami \textit{et al.}'s construction~\cite{takami_Wigner_2001}, after matching the phase and label conventions as in \Cref{eq_takami_phase_dictionary}. Our treatment also determines the topology of this family, making explicit the continuous and discrete freedom that remains after imposing these marginals. 

In their Sec.~3, Chaturvedi \textit{et al.}~\cite{chaturvedi_Wigner_2010} obtain sign choices in a displacement expansion subject to Hermiticity and canonical marginal constraints [their Eqs.~(24), (27)--(28) and (35)]. After matching operator conventions, these sign-valued choices lie within the canonical-marginal family classified here; their subsequent analysis imposes further isotropic-line marginal requirements. 

For a single qudit with ${d\geq2}$, the broader Pauli-covariant operator-basis expansion of Raussendorf \textit{et al.}~\cite{raussendorf_Role_2023} contains the present \criterion{\mathrm{A}1}--\criterion{\mathrm{A}4} family after matching label and phase conventions. This comparison does not establish a correspondence between our stencil phase data and their cohomological obstruction classes; determining whether such a correspondence exists is left to future work.

The admissible base-cell phase data ${\nu}$ also determine the $M$-WHDOs through a rephasing of the base-cell restriction of the doubled WHDO family. At the operator level, this rephasing leaves the commutation law unchanged, while the ordered-product multiplier is modified by the stencil-dependent coboundary ${D_\nu}$; the carry phase ${\theta_c}$ remains fixed by the doubled-WHDO convention. Thus, the continuous and discrete freedom classified here remains even after fixing the Weyl commutation law. The order restriction shows why this distinction matters: it can narrow the available representations by constraining phase choices that the common commutation law leaves undetermined.

It is also interesting to view this topological classification through the associated $M$-characteristic functions. Since the SDFT is an invertible linear map, the families of DWF and characteristic-function representation maps are homeomorphic under their usual finite-dimensional topologies. The same classification therefore applies to both descriptions. For a fixed operator, \Cref{eq_M_characteristic_function_change_of_stencil_rephasing} shows that changing the valid projected stencil preserves all characteristic-function magnitudes and zeros, while nonzero values may acquire different phases. This gives a concrete interpretation of the representation freedom classified here: changing the convention can change how the symplectic-Fourier components combine to form the DWF, while preserving their magnitudes. These preserved magnitudes and zeros therefore provide convention-independent information about the represented operator within the \criterion{\mathrm{A}1}--\criterion{\mathrm{A}4} family.

A further question is how much of this freedom actually changes the phase-space description of a particular operator. If two admissible conventions differ only through rephasings at characteristic-function zeros, they give identical characteristic-function values, and hence identical DWF values, for that operator despite defining distinct representation conventions. Thus, the nonzero entries, together with the admissibility constraints, determine which independent phase choices remain visible in its representation. The SDFT makes this distinction equivalent in the DWF and characteristic-function descriptions.

Another question is how to use this freedom to select a representation suited to a particular task. For example, which admissible convention minimizes the total negative DWF weight of a given state while retaining specified marginal or displacement-order requirements? The characteristic-function viewpoint turns this into a constrained choice of phases with fixed magnitudes. The classification separates the continuous phase choices from the remaining discrete choices, while imposing the order-${d}$ restriction makes the search finite. Investigating this selection problem is left to future work.

A further direction is to extend the classification by relaxing selected validity criteria. Our previous work~\cite{antonopoulos_Grand_2025} discussed relaxing Hermiticity to include discrete Kirkwood--Dirac quasidistributions. Here, Hermiticity gives rise to twisted oddness, relating opposite base-cell labels and restricting non-origin self-negative points to binary choices. How do the independent parameters and topology change when this requirement is relaxed? This would distinguish features specific to Hermiticity from those that persist in a broader stencil framework.

The central message is that the nonuniqueness of discrete Wigner representations has a structure that can be determined explicitly. Within the labeled single-qudit \criterion{\mathrm{A}1}--\criterion{\mathrm{A}4} class, the continuous and discrete freedoms describe the choices that validity leaves open, while additional requirements narrow those choices. This gives a basis for approaching the choice of a DWF through the properties it should satisfy and the freedom those properties leave. What might otherwise be a choice among separate constructions becomes a question about a known parameter space.

\begin{acknowledgments}
\textbf{Author contributions.}
The authors developed the research and mathematical content---including the results, derivations, proofs and substantive discussion---and wrote the original manuscript text. L.K.A.\ led the mathematical development, carried out derivations and proofs, wrote the original manuscript, and led the subsequent revisions and editing. N.C.M.\ supervised the research, jointly developed aspects of the mathematical results with L.K.A., reviewed the manuscript, provided scientific and editorial feedback through multiple revisions, and made additional manuscript edits. Both authors reviewed and approved the final manuscript.

\textbf{AI use.}
OpenAI's ChatGPT (GPT-5.6 Sol, GPT-6 Sol, GPT-6 Astra and GPT-6 Pro) was used as a research assistant for literature searching and contextualization, checking mathematical derivations and proofs, assessing claims and assisting with manuscript editing. The authors reviewed the manuscript and take full responsibility for its content, accuracy and integrity.

\textbf{Acknowledgments and funding.}
We thank Nicholas Funai and Dominic G. Lewis for discussions. This work was supported by the Australian Research Council~(ARC) Centre of Excellence for Quantum Computation and Communication Technology (CQC2T), project no.\ CE170100026, and by the ARC Centre of Excellence for Quantum Computer Performance and Integration (QPI), project no.\ CE260100009. NCM acknowledges support from the ARC Future Fellowship scheme, project no.\ FT230100571.
\end{acknowledgments}

\bibliography{DWF_refs}

\end{document}